\documentclass{article}

\usepackage{microtype}
\usepackage{graphicx}
\usepackage{subcaption}

\usepackage{booktabs} 

\usepackage{hyperref}

\usepackage[accepted]{icml2019}

\usepackage[nomargin,inline,marginclue,draft]{fixme}
\usepackage{amsmath,amsthm,amssymb}
\usepackage{thmtools, thm-restate}
\usepackage{threecolumnalgorithm}
\usepackage{xspace}
\usepackage{listings}

\newtheorem{theorem}{Theorem}[section]
\newtheorem{corollary}[theorem]{Corollary}

\newtheorem{lemma}{Lemma}[section]

\theoremstyle{remark}
\newtheorem{remark}[theorem]{Remark}

\def\[#1\]{\begin{align}#1\end{align}}		
\def\(#1\){\begin{align*}#1\end{align*}} 	

\DeclareMathOperator*{\argmax}{arg\,max} 
\DeclareMathOperator*{\argmin}{arg\,min} 

\newcommand{\phip}{\vec\phi'}
\newcommand{\phipp}{\vec\phi''}
\newcommand{\phippp}{ \vec\phi^{\prime\prime\prime}}

\renewcommand{\algorithmiccomment}[1]{// \textit{#1}}

\allowdisplaybreaks[3]

\usepackage{bm,bbm}
\usepackage[arxiv]{optional}

\usepackage{relsize}
\usepackage{wrapfig}
\usepackage[capitalize]{cleveref}

\crefname{equation}{Eq.}{Eqs.}
\crefname{lemma}{Lemma}{Lemmas}
\crefname{prop}{Proposition}{Propositions}
\crefname{corollary}{Corollary}{Corollaries}
\crefname{theorem}{Theorem}{Theorems}
\crefname{alg}{Algorithm}{Algorithms}

\newcommand{\R}{\mathbb{R}}
\newcommand{\E}{\mathbb{E}}

\newcommand{\pr}{\mathbb{P}}

\newcommand{\convP}{\overset{P}{\to}}

\newcommand{\iid}{\textrm{i.i.d.}\xspace}
\newcommand{\dist}{\sim}
\newcommand{\distiid}{\overset{\textrm{\tiny\iid}}{\dist}}
\newcommand{\distind}{\overset{\textrm{\tiny\textrm{indep}}}{\dist}}

\newcommand{\inv}{^{-1}}
\newcommand{\diag}{\mathrm{diag}}
\newcommand{\trace}{\mathrm{tr}}
\newcommand{\methodname}{LR-GLM\xspace}
\newcommand{\methodnamelap}{LR-Laplace\xspace}

\icmltitlerunning{\methodname: fast GLMs with low-rank approximations}

\newcommand{\abovecapvspace}{\vspace{-0.1in}}
\newcommand{\belowcapvspace}{\vspace{-0.1in}}

\begin{document}

\twocolumn[
\icmltitle{\methodname: High-Dimensional Bayesian Inference Using Low-Rank Data Approximations}



\icmlsetsymbol{equal}{*}

\begin{icmlauthorlist}
\icmlauthor{Brian L.~Trippe}{MIT}
\icmlauthor{Jonathan H.~Huggins}{Harvard}
\icmlauthor{Raj Agrawal}{MIT}
\icmlauthor{Tamara Broderick}{MIT}
\end{icmlauthorlist}

\icmlaffiliation{MIT}{Computer Science and Artificial Intelligence Laboratory, Massachusetts Institute of Technology, Cambridge, MA}
\icmlaffiliation{Harvard}{Department of Biostatistics, Harvard, Cambridge, MA}

\icmlcorrespondingauthor{Brian L.~Trippe}{btrippe@mit.edu}

\icmlkeywords{Machine Learning, ICML}

\vskip 0.3in
]



\printAffiliationsAndNotice{}  

\begin{abstract}

Due to the ease of modern data collection, applied statisticians often have access to a large set of covariates that they wish to relate to some observed outcome. 
Generalized linear models (GLMs) offer a particularly interpretable framework for such an analysis. 
In these high-dimensional problems, the number of covariates is often large relative to the number of observations, so we face non-trivial inferential uncertainty; 
a Bayesian approach allows coherent quantification of this uncertainty. 
Unfortunately, existing methods for Bayesian inference in GLMs require running times roughly cubic in parameter dimension, 
and so are limited to settings with at most tens of thousand parameters. 
We propose to reduce time and memory costs with a low-rank approximation of the data in an approach we call \methodname.
When used with the Laplace approximation or Markov chain Monte Carlo, \methodname provides a full Bayesian posterior approximation 
and admits running times reduced by a full factor of the parameter dimension. 
We rigorously establish the quality of our approximation and show how the choice of rank allows a tunable 
computational--statistical trade-off. 
Experiments support our theory and demonstrate the efficacy of \methodname on real large-scale datasets.

\end{abstract}

\section{Introduction}\label{sec:intro} 

Scientists, engineers, and social scientists are often interested in characterizing the relationship between an outcome and a set of covariates, rather than purely optimizing predictive accuracy.  
For example, a biologist may wish to understand the effect of natural genetic variation on the presence of a disease or a medical practitioner may wish to understand the effect of a patient's history
on their future health.  
In these applications and countless others, the relative ease of modern data collection methods often yields particularly large sets of covariates for data analysts to study.
While these rich data should ultimately aid understanding, they pose a number of practical challenges for data analysis.
One challenge is how to discover interpretable relationships between the covariates and the outcome. 
Generalized linear models (GLMs) are widely used in part because they provide such interpretability -- as well as the flexibility to accommodate a variety of different 
outcome types (including binary, count, and heavy-tailed responses). 
A second challenge is that, unless the number of data points is substantially larger than the number of covariates, there is likely to be non-trivial uncertainty about 
these relationships. 

A Bayesian approach to GLM inference provides the desired coherent uncertainty quantification as well as favorable calibration properties \citep[Theorem 1]{Dawid1982}.
Bayesian methods additionally provide the ability to improve inference by incorporating expert information and sharing power across experiments.
Using Bayesian GLMs leads to computational challenges, however. 
Even when the Bayesian posterior can be computed exactly, conjugate inference costs $O(N^2 D)$ in the case of $D \gg N$. And most models are sufficiently complex as to require expensive approximations.

In this work, we propose to reduce the effective dimensionality of the feature set as a pre-processing step to speed up Bayesian inference, while still performing inference in the original parameter space;
in particular, we show that low-rank descriptions of the data permit fast Markov chain Monte Carlo (MCMC) samplers and Laplace approximations of the Bayesian posterior for the full feature set.
We motivate our proposal with a conjugate linear regression analysis in the case where the data are exactly low-rank.
When the data are merely approximately low-rank, our proposal is an approximation.
Through both theory and experiments, we demonstrate that low-rank data approximations provide a number of properties that are desirable in an efficient posterior approximation method:
(1) \emph{soundness:} 
our approximations admit error bounds directly on the quantities that practitioners report as well as practical interpretations of those bounds;
(2) \emph{tunability:} the choice of the rank of the approximation defines a tunable trade-off between the computational demands of inference and statistical precision; and
(3) \emph{conservativeness:}  our approximation never reports less uncertainty than the exact posterior, where uncertainty is quantified via either posterior variance or information entropy.
Together, these properties allow a practitioner to choose how much information to extract from the data on the basis of computational resources while being able to confidently trust the conclusions of their analysis.


\section{Bayesian inference in GLMs}\label{sec:background}

Suppose we have $N$ data points $\{(x_n, y_n)\}_{n=1}^{N}$. 
We collect our covariates, where $x_n$ has dimension $D$, in the design matrix $X \in \R^{N \times D}$ and our responses in the column vector $Y \in \R^N$. 
Let $\beta \in \R^{D}$ be an unknown parameter characterizing the relationship between the covariates and the response for each data point.
In particular, we take $\beta$ to parameterize a GLM likelihood $p(Y \mid X, \beta) = p(Y \mid X \beta)$. 
That is, $\beta_d$ describes the effect size of the $d$th covariate 
(e.g., the influence of a non-reference allele on an individual's height in a genomic association study). 
Completing our Bayesian model specification, we assume a prior $p(\beta)$, which describes our knowledge of $\beta$ before seeing data. 
Bayes' theorem gives the Bayesian posterior  $p(\beta \mid Y, X) = p(\beta) p(Y \mid X \beta) / \int p(\beta') p(Y \mid X \beta') d\beta'$,
which captures the updated state of our knowledge after observing data. 
We often summarize the $\beta$ posterior via its mean and covariance.
In all but the simplest settings, though, computing these posterior summaries is analytically intractable, and these quantities must be approximated.

\begin{table}[t]
    \caption{Time complexities of naive inference and \methodname with a rank $M$ approximation when $D \ge N$.}\label{table:time_complexities}
\vspace{-0.15in}
\begin{center}
\begin{small}
\begin{sc}
\begin{tabular}{lcccr}
\toprule
Method & Naive & LR-GLM & speedup \\
\midrule
    Laplace              & $O(N^2 D)$ & $O(NDM)$ & ${N}/{M}$ \\
    MCMC (iter.) & $O(ND)$ & $O(NM + DM)$ & ${N}/{M}$ \\
\bottomrule
\end{tabular}
\end{sc}
\end{small}
\end{center}
\vspace{-0.25in}
\end{table}

\textbf{Related work.}
In the setting of large $D$ and large $N$, existing Bayesian inference methods for GLMs may lead to unfavorable trade-offs between accuracy and computation; see \Cref{sec:related_work} for further discussion. 
While Markov chain Monte Carlo (MCMC) can approximate Bayesian GLM posteriors arbitrarily well given enough time, standard methods can be slow, with $O(DN)$ time per likelihood evaluation.
Moreover, in practice, mixing time may scale poorly with dimension and sample size; algorithms thus require many iterations and hence many likelihood evaluations.
Subsampling MCMC methods can speed up inference, but they are effective only with tall data~\citep[$D \ll N$;][]{bardenet2017markov}. 

An alternative to MCMC is to use a deterministic approximation such as the Laplace approximation \citep[Chap. 4.4]{Bishop2006}, integrated nested Laplace 
approximation \citep{rue2009approximate}, variational Bayes \citep[VB;][]{blei2017variational}, or an alternative likelihood
approximation \citep{Huggins2017a,campbell2017automated,huggins2016coresets}.
However these methods are computationally efficient only when $D \ll N$ (and in some cases also when $N \ll D$). 
For example, the Laplace approximation requires inverting the Hessian, which uses $O(\min(N,D)ND)$ time (\Cref{sec:fast_inversions}). 
Improving computational tractability by, for example, using a mean field approximation with VB or a factorized Laplace approximation can produce
substantial bias and uncertainty underestimation \citep{mackay2003information,turner2011two}.

A number of papers have explored using random projections and low-rank approximations in both Bayesian~\citep{lee2013bayesian,spantini2015optimal,guhaniyogi2015bayesian,geppert2017random}
and non-Bayesian~\citep{zhang2014random,wang2017sketching} settings. 
The Bayesian approaches have a variety of limitations. E.g., \citet{lee2013bayesian,geppert2017random,spantini2015optimal} give results only for certain conjugate Gaussian models. And 
\citet{guhaniyogi2015bayesian} provide asymptotic guarantees for prediction but do not address parameter estimation.

See \Cref{sec:results} for a demonstration of the empirical disadvantages of mean field VB, factored Laplace, and random projections in posterior inference.


\section{\methodname}\label{sec:proposal}

\begin{figure*}[!ht]
    \centering
    {\includegraphics[width=0.98\linewidth]{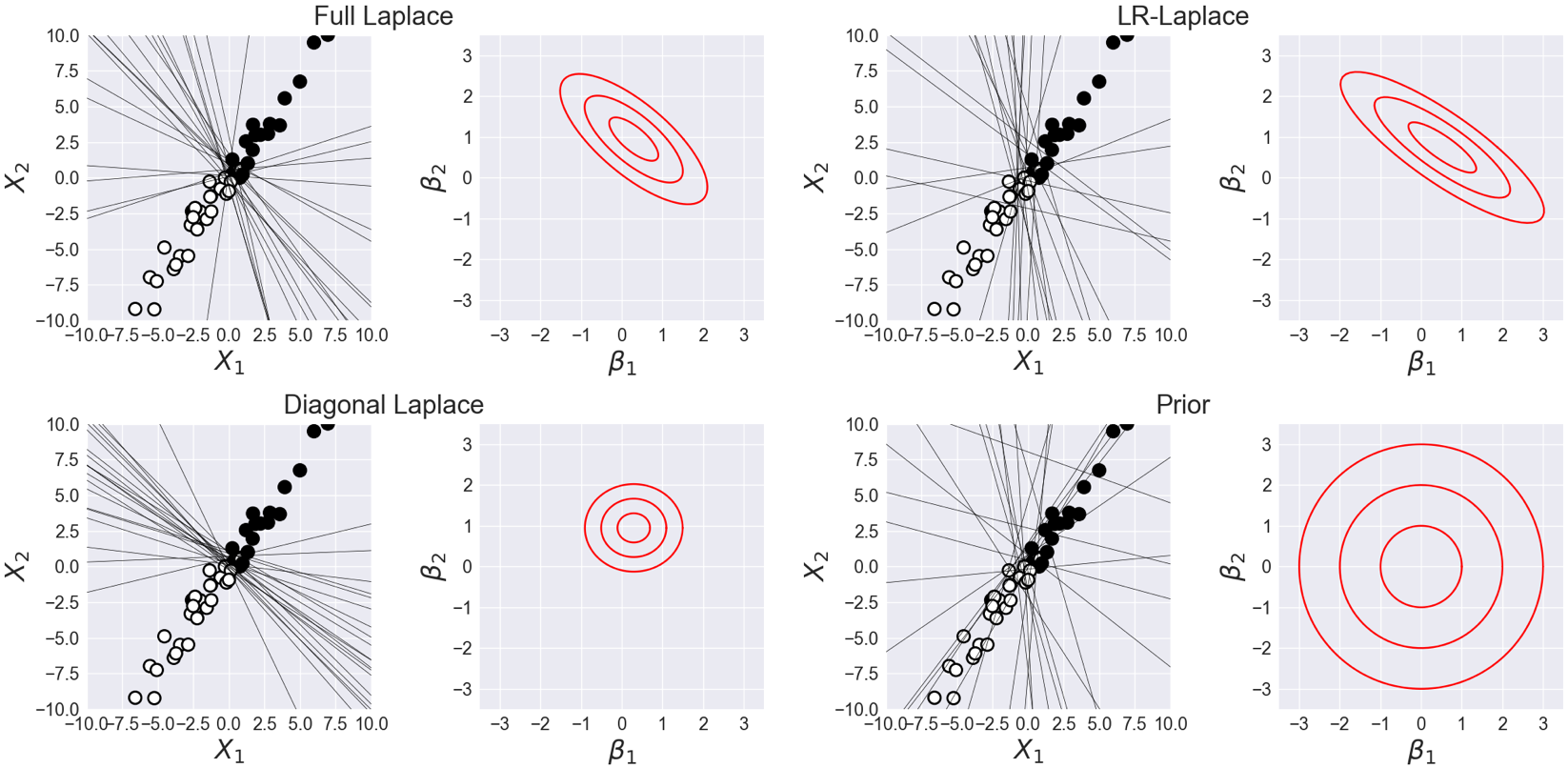}}
    \abovecapvspace
    {\caption{\methodnamelap with a rank-$1$ data approximation closely matches the Bayesian posterior of a toy logistic regression model.
In each pair of plots, the left panel depicts the same 2-dimensional dataset with points in two classes (black and white dots) and decision boundaries (black lines) separating the two classes, which are sampled from the given posterior approximation (see title for each pair).
In the right panel, the red contours represent the marginal posterior approximation of the parameter $\beta$ (a bias parameter is integrated out).}\label{fig:comparison_toy_problem}}
    \belowcapvspace
\end{figure*}

The intuition for our low-rank GLM (\methodname) approach is as follows. Supervised learning problems in high-dimensional settings often exhibit strongly correlated covariates \citep{udell2019nice}. In these cases, the data may provide little information about the parameter along certain directions of parameter space.
This observation suggests the following procedure: first identify a relatively lower-dimensional subspace within which the data most directly inform the posterior, and then perform the data-dependent computations of posterior inference (only) within this subspace, at lower computational expense.
In the context of GLMs with Gaussian priors, the singular value decomposition (SVD) of the design matrix $X$ provides a natural and effective mechanism for identifying a subspace.
We will see that this perspective gives rise to simple, efficient, and accurate approximate inference procedures.
In models with non-Gaussian priors the approximation enables more efficient inference by facilitating faster likelihood evaluations.

Formally, the first step of \methodname is to choose an integer $M$ such that $0 < M < D$. For any real design matrix $X$, its SVD exists and may be written as
$$ 
X^\top  = U\diag(\mathbf{\lambda}) V^\top  + \bar U \diag(\mathbf{\bar \lambda}) \bar V^\top,
$$
where $U\in \R^{D \times M}, \bar U \in \R^{D \times (D-M)}, V\in \R^{N \times M}$, and $\bar V \in \R^{N \times (D-M)}$ are matrices of orthonormal rows, and $\mathbf{\lambda} \in \R^{M}$ and $\mathbf{\bar \lambda} \in \R^{D-M}$ are vectors of non-increasing singular values $\lambda_1 \geq \cdots \geq \lambda_M \geq \bar \lambda_1 \geq \cdots \geq \bar \lambda_{D-M}\geq 0$.
We replace $X$ with the low-rank approximation $XUU^\top $. Note that the resulting posterior approximation $\tilde p(\beta\mid X, Y)$ is still a distribution over the full $D$-dimensional $\beta$ vector: 
\begin{align}\label{eqn:approximation}
\tilde p(\beta\mid X, Y) :=  \frac{p(\beta) p(Y \mid XUU^\top  \beta)}{\int p(\beta') p(Y \mid XUU^\top  \beta') d\beta'}
\end{align}

In this way, we cast low-rank data approximations for approximate Bayesian inference as a likelihood approximation.
This perspective facilitates our analysis of posterior approximation quality and provides the flexibility either to use the likelihood approximation in an otherwise exact MCMC algorithm or to make additional fast approximations such as the Laplace approximation.

We let \emph{\methodnamelap} denote the combination of \methodname and the Laplace approximation.
\Cref{fig:comparison_toy_problem} illustrates \methodnamelap on a toy problem and compares it to full Laplace, the prior, and diagonal Laplace.
Diagonal Laplace refers to a factorized Laplace approximation in which the Hessian of the log posterior is approximated with only its diagonal.
 While this example captures some of the essence of our proposed approach, we emphasize that our focus in this paper is on problems that are high-dimensional.

\section{Low-rank data approximations for conjugate Gaussian regression}\label{sec:conjugate_case}

We now consider the quality of approximate Bayesian inference using our \methodname approach in the case of conjugate Gaussian regression. 
We start by assuming that the data is exactly low rank since it most cleanly illustrates the computational gains from \methodname. 
We then move on to the case of conjugate regression with approximately low-rank data and rigorously characterize the quality of our approximation via interpretable error bounds. 
We consider non-conjugate GLMs in \Cref{sec:non_conjugate}. We defer all proofs to the Appendix.

\subsection{Conjugate regression with exactly low-rank data}\label{sec:conj_low_rank}

Classic linear regression fits into our GLM likelihood framework with $p(Y | X, \beta) = \mathcal{N}(Y | X \beta, (\tau I_{N})\inv)$, where $\tau > 0$ is the precision and $I_{N}$ is the identity matrix of size $N$. For the conjugate prior $p(\beta) = \mathcal{N}(\beta | 0, \Sigma_{\beta})$, we can write the posterior in closed form: 
$
p(\beta | Y, X) = \mathcal{N}(\beta | \mu_N, \Sigma_N),
$
where 
$
\Sigma_N := ( \Sigma_{\beta}\inv + \tau X^{T} X)\inv \text{ and } \mu_N := \tau \Sigma_N X^\top Y.
$

While conjugacy avoids the cost of approximating Bayesian inference, it does not avoid the often prohibitive $O(ND^{2}+D^{3})$ cost of calculating $\Sigma_N$ (which requires computing and then inverting $\Sigma_N\inv$) and the $O(D^2)$ memory demands of storing it.
In the $N\ll D$ setting, these costs can be mitigated by using the Woodbury formula to obtain $\mu_N$ and $\Sigma_N$ in $O(N^2D)$ time with $O(ND)$ memory (\Cref{sec:fast_inversions}).
But this alternative becomes computationally prohibitive as well when both $N$ and $D$ are large (e.g., $D\approx N >{20,000}$).

Now suppose that $X$ is rank $M \ll  \min(D, N)$ and can therefore be written as $X = X U U^{T}$ exactly, 
where $U \in \R^{D \times M}$ denotes the top $M$ right singular vectors of $X$.
Then, if $\Sigma_{\beta} = \sigma^{2}_{\beta} I_{D}$ and $1_M$ is the ones vector of length $M$, we can write (see \Cref{sec:low_rank_blr_proof} for details)
\begin{align}
\hspace{-0.5cm}
\begin{split}
    \Sigma_N = \sigma_\beta^2 &\left\{
		I - U\diag \left(
			\frac{\tau \lambda \odot \lambda}{\sigma_\beta^{-2} 1_M + \tau \lambda \odot \lambda}
			\right)
		U^\top  \right\} \\
	\quad \mathrm{and} \quad
        \mu_N &= U\diag \Bigg( \frac{\tau{\lambda}}{  \sigma_\beta^{-2} 1_M + \tau \lambda \odot \lambda}\Bigg) V^\top Y, \label{eqn:low_rank_blr}
\end{split}  
\end{align}
where multiplication ($\odot$) and division in the diag input are component-wise across the vector $\lambda$. \Cref{eqn:low_rank_blr} provides a more computationally efficient route to inference. 
The singular vectors in $U$ may be obtained in $O(ND \log M)$ time via a randomized SVD \citep{Halko2009} or in $O(NDM)$ time using more standard 
deterministic methods~\cite{press2007numerical}. 
The bottleneck step is finding $\lambda$ via $\diag(\lambda \odot \lambda)=U^\top X^\top XU$, which can be computed in $O(NDM)$ time. 
As for storage, this approach requires keeping only $U$, $\lambda$, and $V^\top Y$, which takes just $O(MD)$ space.
In sum, utilizing low-rank structure via \Cref{eqn:low_rank_blr} provides an order $\min(N,D)/M$-fold improvement in both time and memory over more naive inference.

\subsection{Conjugate regression with low-rank approximations}

While the case with exactly low-rank data is illustrative, real data are rarely exactly low rank.
So, more generally, \methodname will yield an approximation $\mathcal{N}(\beta | \tilde{\mu}_N, \tilde{\Sigma}_{N})$ to the posterior $\mathcal{N}(\beta | \mu_N, \Sigma_N)$, rather than the exact posterior as in \Cref{sec:conj_low_rank}.
We next provide upper bounds on the error from our approximation. Since practitioners typically report posterior means and covariances, we focus on how well \methodname approximates these functionals. 

\begin{theorem} \label{thm:bayes_lin_reg_approx_quality}
For conjugate Bayesian linear regression, the \methodname approximation \Cref{eqn:approximation} satisfies
\begin{align}\label{eqn:blr_mean_error}
 \| \tilde \mu_N- \mu_N \|_2 
 	&\le  \frac{ \bar \lambda_1 \big( \bar \lambda_1 \| \bar U^\top  \tilde \mu_N \|_2 +  \|\bar V^\top  Y\|_2\big)}{\| \tau\Sigma_\beta \|_2\inv+ \bar \lambda_{D-M}^2} \\
\mathrm{and} \quad \Sigma_N\inv - \tilde \Sigma_N\inv
	&=  \tau (X^\top X-UU^\top  X^\top XUU^\top ).
\end{align}
In particular, $\| \Sigma_N\inv - \tilde \Sigma_N\inv \|_2 = \tau \bar \lambda_1^2$. 
\end{theorem}

The major driver of the approximation error of the posterior mean and covariance is $\bar \lambda_1=\|X-XUU^\top\|_{2}$, the largest truncated singular value of $X$. 
This result accords with the intuition that if the data are ``approximately low-rank'' then \methodname should perform well. 

The following corollary shows that the posterior mean estimate is not, in general, consistent for the true parameter.
But it does exhibit reasonable asymptotic behavior. 
In particular, $\tilde \mu_N$ is consistent within the span of $U$ and converges to the \emph{a priori} most probable vector with this characteristic (see the toy example in \Cref{fig:toy_example}).
\begin{corollary} \label{cor:not_consistent}
Suppose $x_{n} \distiid p_{*}$, for some distribution $p^*$, and $y_{n} \mid x_{n} \distind  \mathcal{N}(x_{n}^{\top}\mu_{*}, \tau^{-1})$, for some $\mu_* \in \R^D$. Assume $\E_{p_*}[x_n x_n^\top]$ is nonsingular. Let the columns of $U_*\in \R^{D \times M}$ be the top eigenvectors of $\E_{p_*}[x_n x_n^T]$.
Then $\tilde \mu_N$ converges weakly to the maximum a priori vector $\tilde \mu$ satisfying $U_{*}^\top  \tilde \mu = U_{*}^\top  \mu_{*}$.
\end{corollary}

In the special case that $\Sigma_\beta$ is diagonal this result implies that $\tilde \mu_N \stackrel{p}{\rightarrow} U_{*}U_{*}^\top \mu_{*}$ (\Cref{sec:proof_not_consistent,sec:full-dim-beta-laplace-proof}).  Thus \cref{cor:not_consistent} reflects the intuition that we are not learning anything about the relation between response and covariates in the data directions that we truncate away with our approach.
If the response has little dependence on these directions, $\bar U_{*}\bar U_{*}^\top \mu_{*} = \lim_{N \to \infty} \tilde \mu_{N} - \mu_{*}$ will be small and the error in our approximation will be low (\cref{sec:proof_not_consistent}).
If the response depends heavily on these directions, our error will be higher.
This challenge is ubiquitous in dealing with projections of high-dimensional data.
Indeed, 
we often see explicit assumptions encoding the notion that high-variance directions in $X$ are also highly predictive of the response \citep[see, e.g.,][Theorem 2]{zhang2014random}.

\begin{algorithm*}[!ht]
    \caption{\methodnamelap for Bayesian inference in GLMs with low-rank data approximations and zero-mean prior -- with computation costs. See \Cref{sec:lr_laplace_general} for the general algorithm.}\label{alg:fast_laplace}
\begin{algorithmic}[1]
    \InitThreeCols

    \Phase {{\bfseries Input:} prior $p(\beta)=\mathcal{N}(\mathbf{0}, \Sigma_\beta)$, data $X \in \R^{N,D}$, rank $M\ll D$, GLM mapping $\phi$ with $\phipp$ (see \Cref{eqn:mapping_fcn,sec:fast_laplace_approximations})
    }
    \ThreeHeads{Pseudo-Code}{Time Complexity}{Memory Complexity}
    \Phase{Data preprocessing --- $M$-Truncated SVD }
    \LeftMidRight{\State $U, \diag(\mathbf{\lambda}), V := \operatorname{truncated-SVD}(X^T, M)$}{$O(NDM)$}{$O(NM+ DM)$}
    \vspace{.2cm}
    \Phase{Optimize in projected space and find approximate MAP estimate}
    \LeftMidRight{\State $\gamma_* := \argmax_{\gamma \in \R^M} \sum_{i=1}^N \phi(y_i, x_i U \gamma)-\frac{1}{2} \gamma^\top U^\top\Sigma_\beta U \gamma$}{$O(NM+DM^2)$}{$O(N+M^2)$} \label{line:gamma}
    \LeftMidRight{\State $\hat \mu = U \gamma_* + \bar U \bar U^\top  \Sigma_\beta U(U^\top\Sigma_\beta U)^{-1} \gamma_*$ }{$O(DM)$}{$O(DM)$}\label{line:hat_mu}
    \vspace{.2cm}
    \Phase{Compute approximate posterior covariance}
    \LeftMidRight{\State $W\inv := U^\top  \Sigma_\beta U - (U^\top  X^\top  \diag( \phipp(Y, XUU^\top\hat \mu))XU)\inv$}{$O(NM^2+DM)$}{$O(NM)$} \label{line:W}
    \LeftMidRight{\State $\hat \Sigma := \Sigma_\beta - \Sigma_\beta U W U^\top  \Sigma_\beta$}{0 (see footnote\footnotemark)}{$O(DM)$} \label{line:hat_sigma}
    \vspace{.2cm}
    \Phase{Compute variances and covariances of parameters}
    \LeftMidRight{\State $\mathrm{Var}_{\hat p}(\beta_i)=e_i^\top  \hat \Sigma e_i$}{$O(M^2)$}{$O(DM)$} \label{line:var}
    \LeftMidRight{\State $\mathrm{Cov}_{\hat p}(\beta_i, \beta_j)=e_i^\top  \hat \Sigma e_j$}{$O(M^2)$}{$O(DM)$} \label{line:covar}
  \end{algorithmic}
\end{algorithm*}
\footnotetext{This manipulation is purely symbolic. See \Cref{sec:fl_complexity_proof} for details.}

Our next corollary captures that \methodname never underestimates posterior uncertainty (the \emph{conservativeness} property).
\begin{corollary}\label{cor:conservative}
\methodname approximate posterior uncertainty in any linear combination of parameters is no less than the exact posterior uncertainty. Equivalently, $\tilde \Sigma_N - \Sigma_N$ is positive semi-definite.
\end{corollary}
See \Cref{fig:comparison_toy_problem} for an illustration of this result.
From an approximation perspective, overestimating uncertainty can be seen as preferable to underestimation as it leads to more conservative decision-making.
An alternative perspective is that we actually engender additional uncertainty simply by making an approximation, with more uncertainty for coarser approximations, and we should express that in reporting our inferences.
This behavior stands in sharp contrast to alternative fast approximate inference methods, such as diagonal Laplace approximations (\Cref{sec:bad_marginals}) and variational Bayes \citep{mackay2003information}, which can dramatically underestimate uncertainty.
We further characterize the conservativeness of \methodname in \Cref{cor:information_gain}, which shows that the \methodname posterior never has lower entropy than the exact posterior and quantifies the bits of information lost due to approximation.

\section{Non-conjugate GLMs with approximately low-rank data} \label{sec:non_conjugate}
While the conjugate linear setting facilitates intuition and theory, GLMs are a larger and more broadly useful class of models for which efficient and reliable Bayesian inference is of significant practical concern.  
Assuming conditional independence of the observations given the covariates and parameter, the posterior for a GLM likelihood can be written
\[\label{eqn:mapping_fcn}
\log p(\beta\mid X, Y) = \log p(\beta) + \sum_{n=1}^N \phi(y_n, x_n^\top \beta) + Z
\]
for some real-valued mapping function $\phi$ and log normalizing constant $Z$.
For priors and mapping functions that do not form a conjugate pair, accessing posterior functionals of interest is analytically intractable and requires posterior approximation.
One possibility is to use a Monte Carlo method such as MCMC, which has theoretical guarantees asymptotic in running time but is relatively slow in practice.
The usual alternative is a deterministic approximation such as VB or Laplace. These approximations are typically faster but do not become arbitrarily accurate in the limit of infinite computation.
We next show how \methodname can be applied to facilitate faster MCMC samplers and Laplace approximations for Bayesian GLMs.
We also characterize the additional error introduced to Laplace approximations by low-rank data approximations.

\subsection{\methodname for fast Laplace approximations}\label{sec:fast_laplace_approximations}

The Laplace approximation refers to a Gaussian approximation obtained via a second-order Taylor approximation of the log density. In the Bayesian setting, the Laplace approximation $\bar p(\beta\mid X, Y)$ is typically formed at the maximum a posteriori (MAP) parameter:
$
        \bar p(\beta\mid X, Y) := \mathcal{N}(\beta \mid  \bar \mu, \bar \Sigma),
$
where $\bar\mu := {\textstyle\argmax_\beta} \log p(\beta\mid X, Y)$ and 
$\bar \Sigma\inv := -\nabla_\beta^2 \log p(\beta\mid X, Y)|_{\beta=\bar \mu}$.
When computing and analyzing Laplace approximations for GLMs, we will often refer to vectorized first, second, and third derivatives $\phip, \phipp, \phippp \in \R^{N}$ of the mapping function $\phi$. For $Y, A \in \mathbb{R}^N$, we define
$
\phip(Y, A)_n := \frac{\partial}{\partial a} \phi(Y_n, a)|_{a=A_n}.
$
The higher-order derivative definitions are analogous, with the derivative order of $\frac{\partial}{\partial a}$ increased commensurately.

Laplace approximations are typically much faster than MCMC for moderate/large $N$ and small $D$, but they become expensive or intractable for large $D$. In particular, they require inverting a $D \times D$ Hessian matrix, which is in general an $O(D^3)$ time operation, and storing the resulting covariance matrix, which requires $O(D^2)$ memory.\footnote{Notably, as in the conjugate setting, an alternative matrix inversion using the Woodbury identity reduces this cost when $N<D$ to $O(N^2D)$ time and $O(ND)$ memory (\Cref{sec:fast_inversions}).}

As in the conjugate case, \methodname permits a faster and more memory-efficient route to inference. Here, we say that the \emph{\methodnamelap approximation}, $\hat p(\beta\mid X, Y) = \mathcal{N}(\beta \mid  \hat \mu, \hat \Sigma)$, denotes the Laplace approximation to the \methodname approximate posterior.
The special case of \methodnamelap with zero-mean prior is given in \Cref{alg:fast_laplace} as it allows us to easily analyze time and memory complexity.  For the more general \methodnamelap algorithm, see \Cref{sec:lr_laplace_general}.

\begin{theorem}\label{thm:laplace_time}
In a GLM with a zero-mean, structured-Gaussian prior\footnote{For example (banded) diagonal or diagonal plus low-rank, such that matrix vector multiplies may be computed in $O(D)$ time.} 
and a log-concave likelihood,\footnote{This property is standard for common GLMs such as logistic and Poisson regression.}
the rank-$M$ \methodnamelap approximation may be computed via \Cref{alg:fast_laplace} in $O(NDM)$ time with $O(DM+NM)$ memory.
Furthermore, any posterior covariance entry can be computed in $O(M^2)$ time.
\end{theorem}

\Cref{alg:fast_laplace} consists of three phases: (1) computation of the $M$-truncated SVD of $X^\top$; (2) MAP optimization to find $\hat \mu$; and (3) estimation of $\hat \Sigma$.
In the second phase we are able to efficiently compute $\hat \mu$ by first solving a lower-dimensional optimization for the quantity $\gamma_*\in \R^M$ (\cref{line:gamma}), from which $\hat \mu$ is available analytically.
Notably, in the common case that $p(\beta)$ is isotropic Gaussian, the expression for $\hat \mu$ reduces to $U\gamma_*$ and the full time complexity of MAP estimation is $O(NM+DM)$.
Though computing the covariance for each pair of parameters and storing $\hat \Sigma$ explicitly would of course require a potentially unacceptable $O(D^2)$ storage, 
the output of \Cref{alg:fast_laplace} is smaller and enables arbitrary parameter variances and covariances to be computed in $O(M^2)$ time.  See \Cref{sec:fl_complexity_proof} for additional details.

\subsection{Accuracy of the \methodnamelap approximation}
We now consider the quality of the \methodnamelap approximate posterior relative to the usual Laplace approximation. 
 Our first result concerns the difference of the posterior means

\begin{theorem}[Non-asymptotic]\label{thm:glm_mean_bound}
In a generalized linear model with an $\alpha$--strongly log concave posterior, the exact and approximate MAP values, $\hat \mu = \argmax_\beta \tilde p(\beta\mid X, Y)$ and $\bar \mu = \argmax_\beta p(\beta\mid X, Y)$, satisfy
$$
\| \hat \mu  - \bar \mu \|_2 \le \frac{\bar \lambda_1 \big( \|\phip(Y, X\hat \mu)\|_2 + \lambda_1 \| \bar U^\top\hat \mu\|_2 \| \phipp(Y, A)\|_\infty\big)}{\alpha}
$$
for some vector $A \in \mathbb{R}^N$ such that $A_n \in [x_n^\top UU^\top \hat \mu, x_n^\top \hat \mu]$.
\end{theorem}
This bound reveals several characteristics of the regimes in which \methodnamelap performs well.  As in conjugate regression, we see that the bound tightens to $0$ as the rank of the approximation increases to capture all of the variance in the covariates and $\bar \lambda_1 \rightarrow 0$. 
\begin{remark}\label{rem:common_glm_derivatives}
For many common GLMs, $\|\phip \|_2$, $\|\phipp \|_\infty$, and $\| \phippp \|_\infty$ are well controlled; see \Cref{sec:log_reg}. $\| \phippp \|_\infty$ appears in an upcoming corollary.
\end{remark}

\begin{remark}
The $\alpha$--strong log concavity of the posterior is satisfied for any strongly log concave prior (e.g., a Gaussian, in which case we have $\alpha \ge \|\Sigma_\beta\|_2\inv$) 
and $\phi(y,\cdot)$ is concave for all $y$.  In this common case, \Cref{thm:glm_mean_bound} provides a computable upper bound on the posterior mean error.
\end{remark}

\begin{remark}
 In contrast to the conjugate case (\Cref{cor:not_consistent}), general LR-GLM parameter estimates are not necessarily consistent within the span of the projection. That is, $U^\top \hat \mu_N$  may not converge to $U^\top \beta$ (see \Cref{sec:glm_not_consistent}).  
\end{remark}

We next consider the distance between our approximation and target posterior under a Wasserstein metric \citep{villani2008optimal}. Let $\Gamma(\hat p, \bar p)$ be the set of all couplings of distributions $\hat p$ and $\bar p$, i.e. joint distributions $\gamma(\cdot, \cdot)$ satisfying  $\hat p(\beta) =\int \gamma(\beta, \beta^\prime) d\beta^\prime$ and $\bar p(\beta) =\int \gamma(\beta^\prime, \beta) d\beta^\prime$ for all $\beta$. Then the $2$-Wasserstein distance between $\hat p$ and $\bar p$ is defined
\[
    W_2(\hat p, \bar p) = \inf_{\gamma \in \Gamma(\hat p, \bar p)} \E_{\gamma } [\|\hat \beta-\bar \beta\|^2_2]^{\frac{1}{2}}.
\]
Wasserstein bounds provide tight control of many functionals of interest, such as means, variances, and standard deviations \cite{huggins2018nonasymptotic}. 
For example, if $\xi_{i} \dist q_{i}$ for any distribution $q_{i}~(i=1,2)$, then $|\E[\xi_{1}]-\E[\xi_{2}]| \le W_2(q_{1}, q_{2})$ and 
$|\mathrm{Var}[\xi_{1}]^\frac{1}{2} - \mathrm{Var}[\xi_{2}]^\frac{1}{2}| \le 2 W_2(q_{1},q_{2})$.

We provide a finite-sample upper bound on the 2-Wasserstein distance between the Laplace and \methodnamelap approximations.
In particular, the 2-Wasserstein will decrease to $0$ as the rank of the \methodnamelap approximation increases since the largest truncated singular value $\bar \lambda_1$ will approach zero.  
\begin{corollary}\label{cor:new_glm_w2_bound}
Assume the prior $p(\beta)$ is Gaussian with covariance $\Sigma_\beta$ and the mapping function $\phi(y, a)$ has bounded 2nd and 3rd derivatives with respect to $a$. Take $A$ and $\alpha$ as in \Cref{thm:glm_mean_bound}.
Then $\bar p(\beta)$ and $\hat p(\beta)$ satisfy
\[
    W_2(\hat p, \bar p) \le &\sqrt{2} \bar \lambda_1 \| \bar \Sigma \|_2 
    \Big\{
    c\big[\|\Sigma_\beta\inv\|_2 + (\lambda_1 + \bar \lambda_1)^2 \|\phipp\|_\infty\big] \nonumber \\
    &+(\lambda_1^2 r +  (\bar \lambda_1 +2 \lambda_1)\|\phipp\|_\infty \sqrt{\trace(\hat \Sigma)}
    \Big\},
\]
    where $c := \big(\|\phip(Y, X\hat \mu)\|_2 + \lambda_1 \| \bar U^\top\hat \mu\|_2 \| \phipp(Y, A)\|_\infty\big) / \alpha$  and $r:= \| U^\top \hat \mu\|_\infty \|\phippp\|_\infty + \lambda_1 c\|\phippp\|_\infty$.
\end{corollary}
When combined with \citet[Prop. 6.1]{huggins2018nonasymptotic}, this result guarantees closeness in 2-Wasserstein of \methodnamelap to the exact posterior.

We conclude with a result showing that the error due to the \methodname approximation cannot grow without bound as the sample size increases. 
\begin{theorem}[Asymptotic]\label{thm:asymptotic_glm_mean_bound_concise}
Under mild regularity conditions, the error in the posterior means, $\| \hat \mu_n - \bar \mu_n\|_2$, converges as $n \rightarrow \infty$, and the limit is finite almost surely.
\end{theorem}
For the formal statement see \Cref{thm:asymptotic_glm_mean_bound} in \Cref{sec:glm_bounded_asymptotic_error}.

\subsection{LR-MCMC for faster MCMC in GLMs}\label{sec:fast_mcmc}
\methodnamelap is inappropriate when the posterior is poorly approximated by a Gaussian.
This may be the case, for example, when the posterior is multi-modal, a common characteristic of GLMs with sparse priors.
To remedy this limitation of \methodnamelap, we introduce \emph{LR-MCMC}, a wrapper around the Metropolis--Hastings algorithm using the \methodname approximation.
For a GLM, each full likelihood and gradient computation takes $O(ND)$ time but only $O(NM+DM)$ time with the \methodname approximation, resulting in the same $\mathrm{min}(N, D)/M$-fold speedup obtained by \methodnamelap.
See \Cref{sec:lr_mcmc} for further details on LR-MCMC.

\section{Experiments}\label{sec:results}

We empirically evaluated \methodname on real and synthetic datasets.  For synthetic data experiments, we considered logistic regression with covariates of dimension $D=250$ and $D=500$.  In each replicate, we generated the latent parameter from an isotropic Gaussian prior, $\beta \sim  \mathcal{N}(0, I_D)$, correlated covariates from a multivariate Gaussian, and responses from the logistic regression likelihood (see \Cref{sec:experimental_details} for details). 
We compared to the standard Laplace approximation, the diagonal Laplace approximation, the Laplace approximation with a low-rank data approximation obtained via random projections rather than the SVD (``Random-Laplace''), and mean-field automatic differentiation variational inference in \texttt{Stan} (ADVI-MF).\footnote{We also tested ADVI using a full rank Gaussian approximation but found it to provide near uniformly worse performance compared to ADVI-MF. So we exclude full-rank ADVI from the presented results.}  

\begin{figure*}[!ht]
    \centering
    {\includegraphics[width=1.\linewidth]{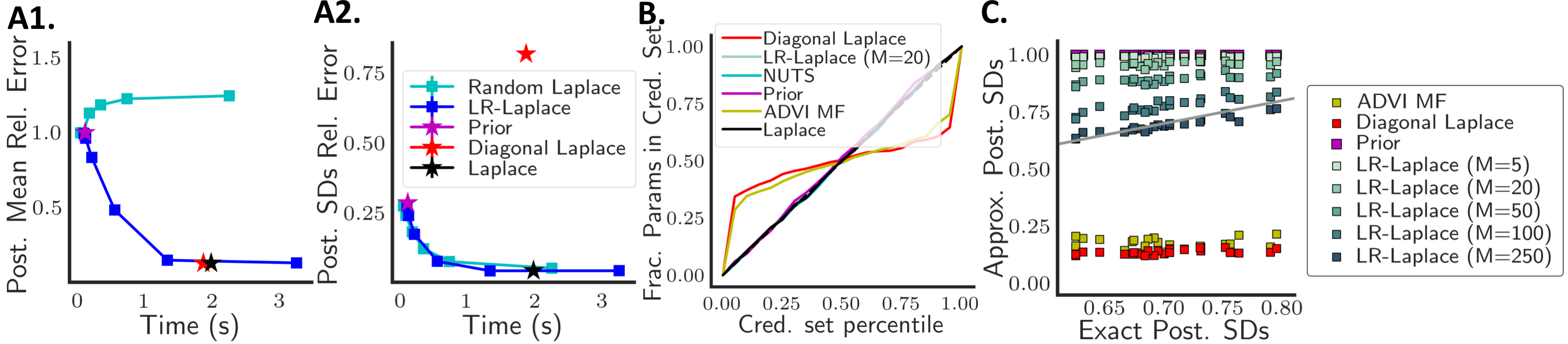}}
\abovecapvspace
    {\caption{\emph{Left}: Error of the approximate posterior (A1.) mean and (A2.) variances relative to ground truth (running NUTS with \texttt{Stan}).  Lower and further left is better.  \emph{Right} (B.): Credible set calibration across all parameters and repeated experiments. (C.): Approximate posterior standard deviations for a subset of parameters.  The grey line reflects zero error.}\label{fig:comp_trade_off_laplace}}
    \belowcapvspace
\end{figure*}

\textbf{Computational--statistical trade-offs.}
\Cref{fig:comp_trade_off_laplace}A shows empirically the 
tunable computational--statistical trade-off offered by varying $M$ in our low-rank data approximation. This plot
depicts the error in posterior mean and variance estimates relative to results from the No-U-Turn Sampler (NUTS) in \texttt{Stan} \citep{hoffman2014no,carpenter2017stan}, which we treat as ground truth. As expected, \methodnamelap with larger $M$ takes longer to run but yields lower errors.  Random-Laplace was usually faster but provided a poor posterior approximation.  Interestingly, the error of the Random-Laplace approximate posterior mean actually increased with the dimension of the projection. We conjecture this behavior may be due to Random-Laplace prioritizing covariate directions that are correlated with directions where the parameter, $\beta$, is large.

We also consider predictive performance via the classification error rate and the average negative log likelihood. In particular, we generated a \emph{test} dataset with covariates drawn from the same distribution as the observed dataset and an \emph{out-of-sample} dataset with covariates drawn from a different distribution (see \Cref{sec:experimental_details}). The computation time vs.\ performance trade-offs, presented in \Cref{fig:pred_error} on the test and out-of-sample datasets, mirror the results for approximating the posterior mean and variances. In this evaluation, correctly accounting for posterior uncertainty appears less important for in-sample prediction. But in the {out-of-sample} case, we see a dramatic difference in negative log likelihood. Notably, ADVI-MF and diagonal Laplace exhibit much worse performance.  These results support the utility of correctly estimating Bayesian uncertainty when making out-of-sample predictions.

\textbf{Conservativeness.}
A benefit of \methodname is that the posterior approximation never underestimates the posterior uncertainty (see \Cref{cor:conservative}). \Cref{fig:comp_trade_off_laplace}C illustrates this property for \methodnamelap applied to logistic regression. When \methodnamelap misestimates posterior variances, it always overestimates. Also, when \methodnamelap misestimates means (\Cref{fig:means_and_variances_laplace}), the estimates shrink closer to the prior mean, zero in this case. These results suggest that \methodname interpolates between the exact posterior and the prior.  Notably, this property is not true of all methods. The diagonal Laplace approximation, by contrast, dramatically underestimates posterior marginal variances (see \Cref{sec:bad_marginals}).

\textbf{Reliability and calibration.}
Bayesian methods enjoy desirable calibration properties under correct model specification. But since \methodnamelap serves as a likelihood approximation, it does not retain this theoretical guarantee.  
Therefore, we assessed its calibration properties empirically by examining the credible sets of both parameters and predictions. We found that the parameter credible sets of \methodnamelap are extremely well calibrated for all values of $M$ between 20 and 400 (\Cref{fig:comp_trade_off_laplace}B and \Cref{fig:credible_set_calibration}). The prediction credible sets were well calibrated for all but the smallest value of $M$ tested ($M=20$); in the $M=20$ case, \methodnamelap yielded under-confident predictions (\Cref{fig:pred_calibration}).
The good calibration of \methodnamelap stood in sharp contrast to the diagonal Laplace approximation and ADVI-MF. Random-Laplace also provided inferior calibration (\Cref{fig:credible_set_calibration,fig:pred_calibration}).

\textbf{\methodname with MCMC and non-Gaussian priors.}
In \Cref{sec:fast_mcmc} we argued that \methodname speeds up MCMC for GLMs by decreasing the cost of likelihood and gradient evaluations in black-box MCMC routines. We first examined LR-MCMC with NUTS using \texttt{Stan} on the same synthetic datasets as we did for \methodnamelap. In \Cref{fig:means_and_variances_HMC,fig:comp_trade_off_hmc}, we see a similar conservativeness and computational--statistical trade-off as for \methodnamelap, and superior performance relative to alternative methods. 

We expect MCMC to yield high-quality posterior approximations across a wider range of models than Laplace approximations.
For example, for multimodal posteriors and other posteriors that deviate substantially from Gaussianity. 
We next demonstrate that LR-MCMC is useful in these more general cases.
In high-dimensional settings, practitioners are often interested in identifying a sparse subset of parameters that significantly influence responses. 
This belief may be incorporated in a Bayesian setting through a sparsity-inducing prior such as the spike and slab prior or the horseshoe \cite{george1993variable,carvalho2009handling}. 
However, posteriors in these cases may be multimodal, and scalable Bayesian inference with such priors is a challenging, active area of research \cite{guan2011bayesian,yang2016computational,johndrow2017scalable}.
To demonstrate the applicability of low-rank data approximations to this setting, we ran NUTS using \texttt{Stan} on a logistic regression model with a regularized horseshoe prior \citep{carvalho2009handling,piironen2017sparsity}.  In \Cref{fig:horseshoe}, we see an attractive trade-off between computational investment and approximation error. For example, we obtained relative mean and standard deviation errors of only about $10^{-2}$ while reducing computation time by a factor of three.

We also applied LR-MCMC to linear regression with the regularized horseshoe prior on a dataset with very correlated covariates and $D=6{,}238$.  However, this sampler exhibited severe mixing problems, both with and without the approximation, as diagnosed by large $\hat R$ values in \texttt{pyStan}.  These issues reflect the innate challenges of high-dimensional Bayesian inference with the horseshoe prior and correlated covariates.

\textbf{Scalability to large-scale real datasets.} 
Finally, we explored the applicability of \methodnamelap to two real, large-scale logistic regression tasks (\Cref{fig:real_data_results}).
The first is the UCI Farm-Ads dataset, which consists of $N=$ {4,143} online advertisements for animal-related topics together with binary labels indicating whether the content provider approved of the ad; there are $D=$ {54,877} bag-of-words features per ad \citep{Dua:2017}.
As with the synthetic datasets, we evaluated the error in the approximations of posterior means and variances.
As a baseline to evaluate this error, we use the usual Laplace approximation because the computational demands of MCMC preclude the possibility of using it as a baseline.

As a second real dataset we evaluated our approach on the Reuters RCV1 text categorization test collection \cite{amini2009learning,chang2011libsvm}.
RCV1 consists of $D =$ {47,236} bag-of-words features for $N=$ {20,241} English documents grouped into two different categories.
We were unable to compare to the full Laplace approximation due to the high-dimensionality, so we used \methodnamelap with $M=$ {20,000} as a baseline.
For both datasets, we find that as we increase the rank of the data approximation, we incur longer running times but reduced errors in posterior means and variances.
Laplace and Diagonal Laplace do not provide the same computation--accuracy trade-off. 

\begin{figure}[!ht]
    \centering
    {\includegraphics[width=1.0\linewidth]{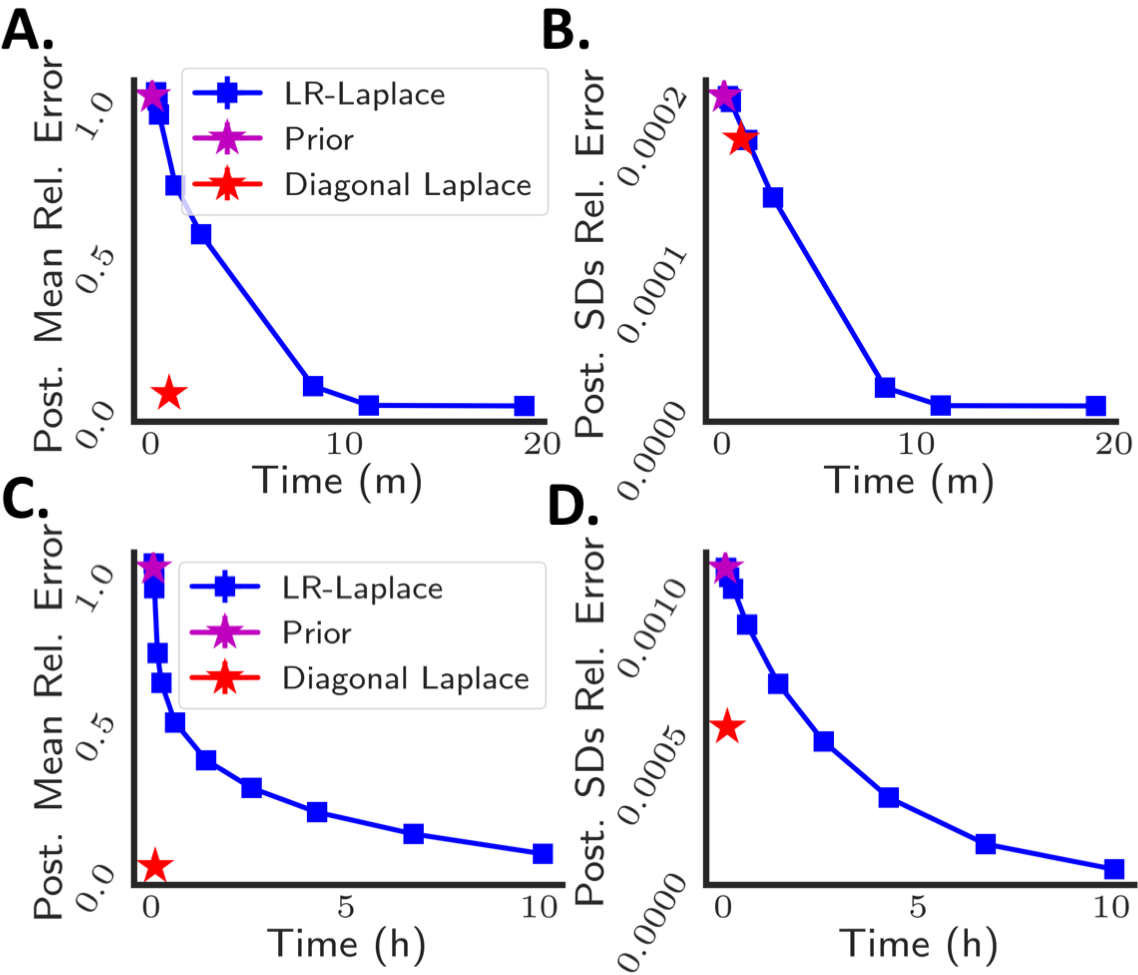}}
\abovecapvspace
    {\caption{\methodnamelap approximation quality on Farm-Ads (top) and RCV-1 (bottom) datasets with varying $M$.
    (A.) Farm-Ads error in the posterior mean and (B.) Farm-Ads error in posterior variances (C.) RCV-1 error in posterior mean and (D.) RCV-1 error in posterior variances.}\label{fig:real_data_results}}
    \belowcapvspace
\end{figure}

\textbf{Choosing $M$.}
Applying \methodname requires choosing the rank $M$ of the low rank approximation.  
As we have shown, this choice characterizes a computational--statistical trade-off whereby larger $M$ leads to linearly larger computational demands, but increases the precision of the approximation.
As a practical rule of thumb, we recommend setting $M$ to be as large as is allowable for the given application without the resulting inference becoming too slow.
For our experiments with \methodnamelap, this limit was $M \approx 20{,}000$. 
For LR-MCMC, the largest manageable choice of $M$ will be problem dependent but will typically be much smaller than {20,000}. 


\section*{Acknowledgments}
This research is supported in part by an NSF CAREER Award, an ARO YIP Award, a Google Faculty Research Award, a Sloan Research Fellowship, and ONR.  BLT is supported by NSF GRFP.


\bibliography{./references}
\bibliographystyle{icml2019}

\opt{arxiv}{\newpage\onecolumn\appendix 
\numberwithin{equation}{section}
\numberwithin{figure}{section}

\section{Additional Experimental Details and Empirical Results}

\subsection{Experimental Details}\label{sec:experimental_details}
For all experiments we sampled $\beta$ from an isotropic Gaussian prior with unit variance.
For all synthetic data results we first generated a design matrix by sampling from a zero-mean Gaussian with diagonal covariance $\Sigma$ with each $\Sigma_{i,i} = 5*1.
05^{-i}$.
We then used a scikit-learn \citep{Pedregosa} implementation of a randomized SVD algorithm due to \citet{Halko2009}, computed from two iterations (i.e., passes through $X$).

To assess robustness, in all experiments we used three or more replicate experiments, defined by independently generated synthetic datasets or train/test splits as well as re-rerunning the randomized truncated SVD.

The performance of the Diagonal Laplace approximation is dependent upon the shape the exact posterior at $\beta^\mathrm{MAP}$. In particular, using a dataset with axis aligned covariance structure gives Diagonal Laplace an unrealistic advantage given that in most real applications we do not believe that low-rank structure will be axis aligned.
As such, for all synthetic data experiments presented, we randomly generated a basis of orthonormal vectors and used this basis to rotate our the design matrix.
This rotation preserves the spectral decay of the data but eliminates the axis alignment of the synthetic data.

In all experiments we consider $N=2{,}500$ training examples.
We obtained results on ``Out of Sample Data'' (in \Cref{fig:pred_error,fig:pred_calibration}) by sampling $X$ from an alternative distribution over covariates.
Specifically, we generated these out-of-sample covariates in the manner described above, but with a different random rotation matrix.

We found MAP estimation using L-BFBS-B to be the most efficient of several available options in the scipy optimize library, and used this method in all MAP estimation and Laplace approximation experiments.

For all Bayesian predictions, we use the probit approximation to the logistic function to enable fast approximation \citep[Chap.\ 4.5]{Bishop2006}.

\subsection{Additional Figures}
In \Cref{fig:pred_error} we present results on prediction performance, in term of classification error, as well as negative log likelihood, reported for ``Training'', ``Test'', and ``Out of Sample Data''.  In \Cref{fig:means_and_variances_laplace} we report the error of LR-Laplace and Random-Laplace relative to NUTS for estimation of posterior means and variances.  We see here that the estimates exhibit behavior increasingly similar to that of the prior as the rank of the approximation, $M$, decreases.  Next, \Cref{fig:means_and_variances_HMC} depicts the same error trends for LR-MCMC using NUTS in \texttt{Stan}.  We report calibration performance of the approximations of interest for credible sets of parameters (\Cref{fig:credible_set_calibration}) as well as for prediction (\Cref{fig:pred_calibration}).

We additionally include results analogous to those in the main text for Laplace approximations using low-rank data approximations to perform faster MCMC using NUTS with \texttt{Stan} \citep{carpenter2017stan}, in \Cref{fig:comp_trade_off_hmc}.  Finally, we also here provide the relative error of posterior mean and standard deviation estimation for logistic regression with a regularized horseshoe prior using the LR-MCMC approximation in \Cref{fig:horseshoe}.  This experiment uses \texttt{Stan} for inference as well.
\begin{figure}[!ht]
    \centering
    {\includegraphics[width=0.90\linewidth]{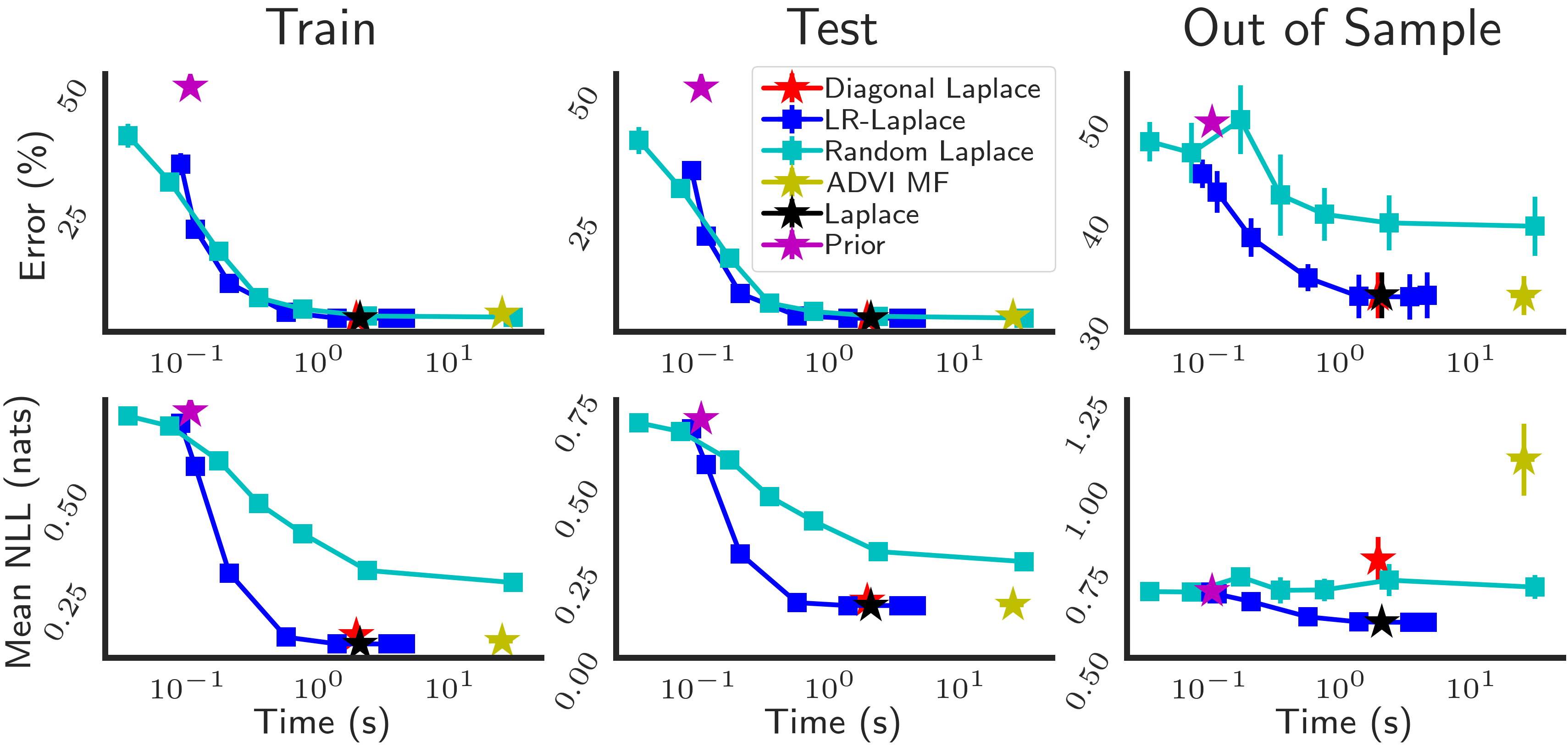}}
\abovecapvspace
    {\caption{Predictive performance of posterior approximations in Bayesian logistic regression in terms of (Top) classification error and (Bottom) average negative log likelihood (NLL) of responses under approximate posterior predictive distributions on (Left) \emph{train}, (Center) \emph{test} and  (Right) \emph{out of sample} datasets.
    Lower is better.}\label{fig:pred_error}}
    \belowcapvspace
\end{figure}

\begin{figure*}[!ht]
    \centering
    \includegraphics[width=0.8\linewidth]{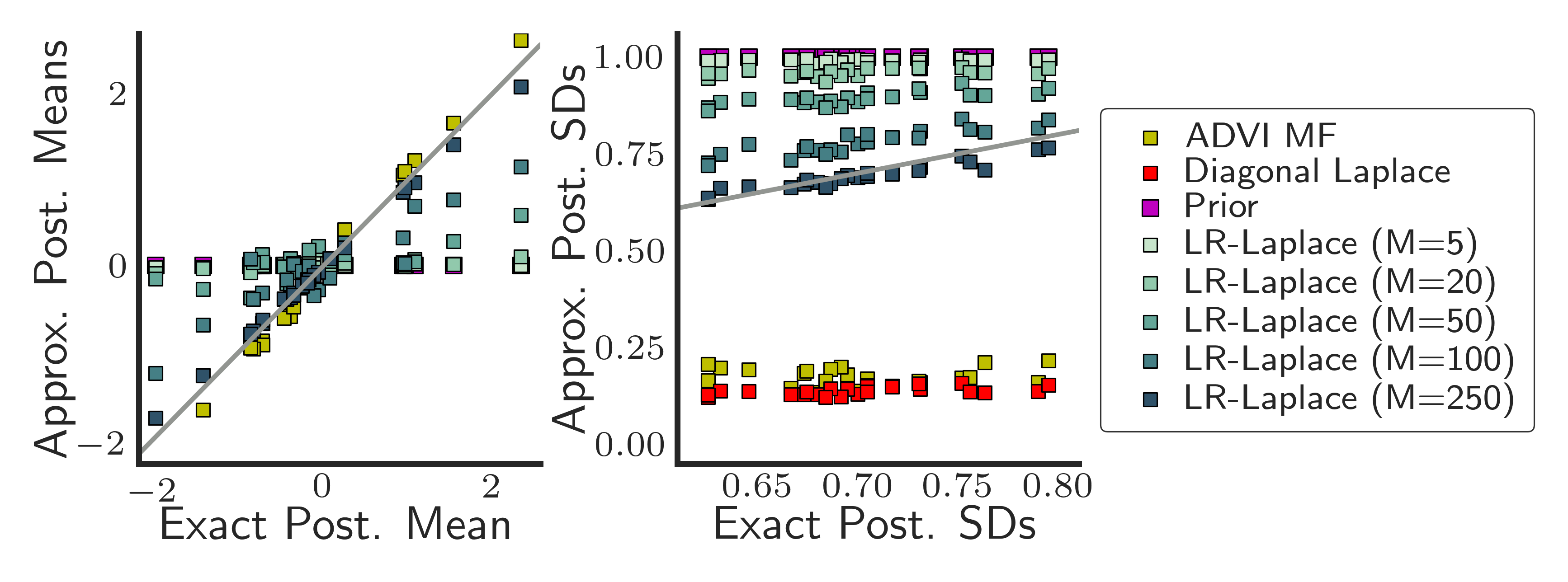}
    \caption{Approximate posterior mean and standard deviation across a parameter subset as $M$ varies. Horizontal axis represents ground truth from running NUTS using \texttt{Stan} without the \methodname approximation.  $D=250$.}\label{fig:means_and_variances_laplace}
\end{figure*}

\begin{figure}[htbp]
\centering
 {\includegraphics[width=.8\linewidth]{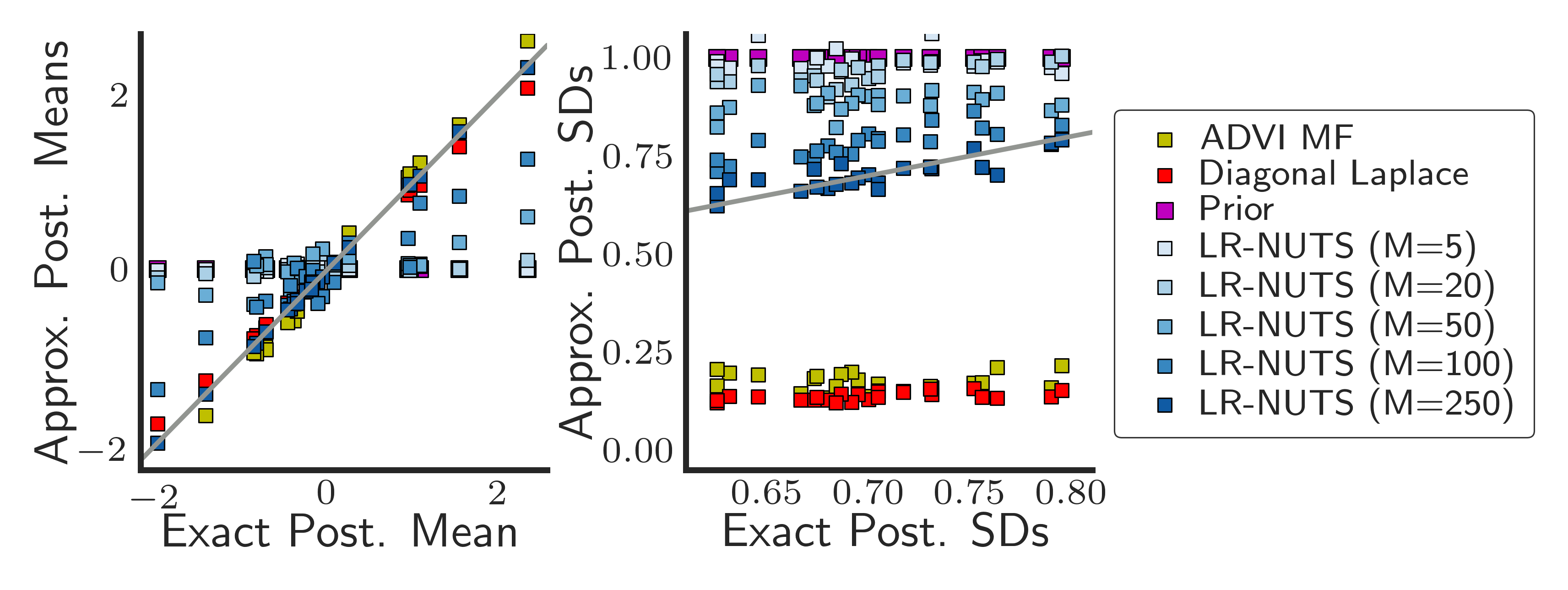}}
    {\caption{This figure is analogous to \Cref{fig:means_and_variances_laplace} but examines the trade-off between computation and accuracy of LR-MCMC using NUTS in \texttt{Stan}. $D=250$.}\label{fig:means_and_variances_HMC}}
\end{figure}

\begin{figure}[!ht]
    \centering
    {\includegraphics[width=0.5\linewidth]{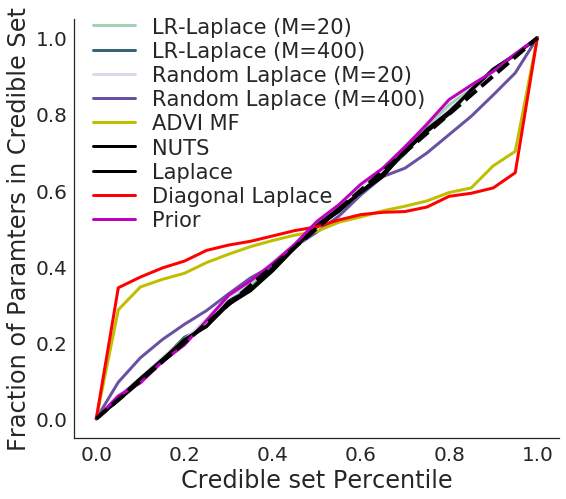}}
\abovecapvspace
    {\caption{Credible set calibration. The fraction of parameters in the credible sets defined by different lower tail intervals as a function of the approximate posterior probability of parameters taking values in that interval. The black dotted line (on the diagonal) reflects perfect calibration.}\label{fig:credible_set_calibration}}
    \belowcapvspace
\end{figure}

\begin{figure*}[!ht]
  \centering
  {\includegraphics[width=0.75\linewidth]{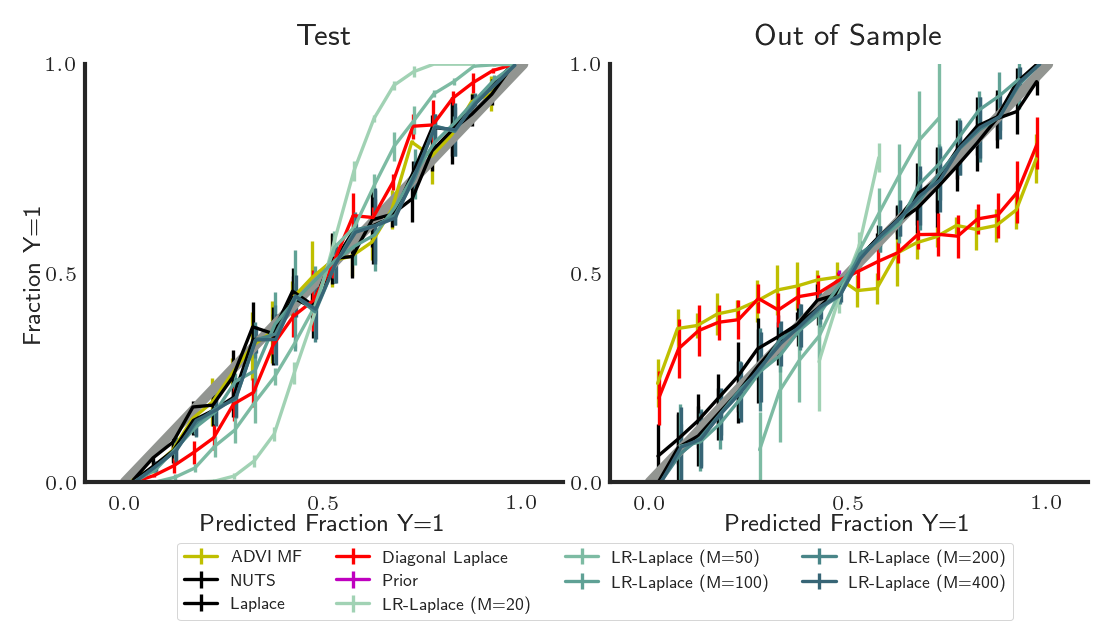}}
  {\includegraphics[width=0.75\linewidth]{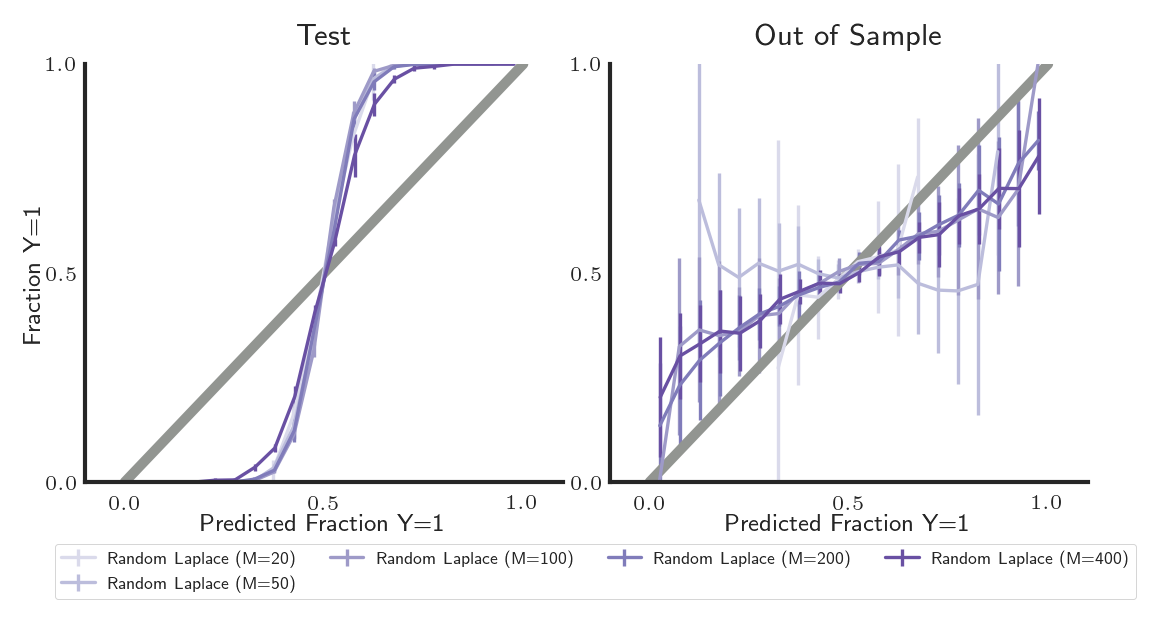}}
  {\caption{Prediction calibration.}\label{fig:pred_calibration}}
\end{figure*}

\begin{figure}[htbp]
\centering
 {\includegraphics[width=.75\linewidth]{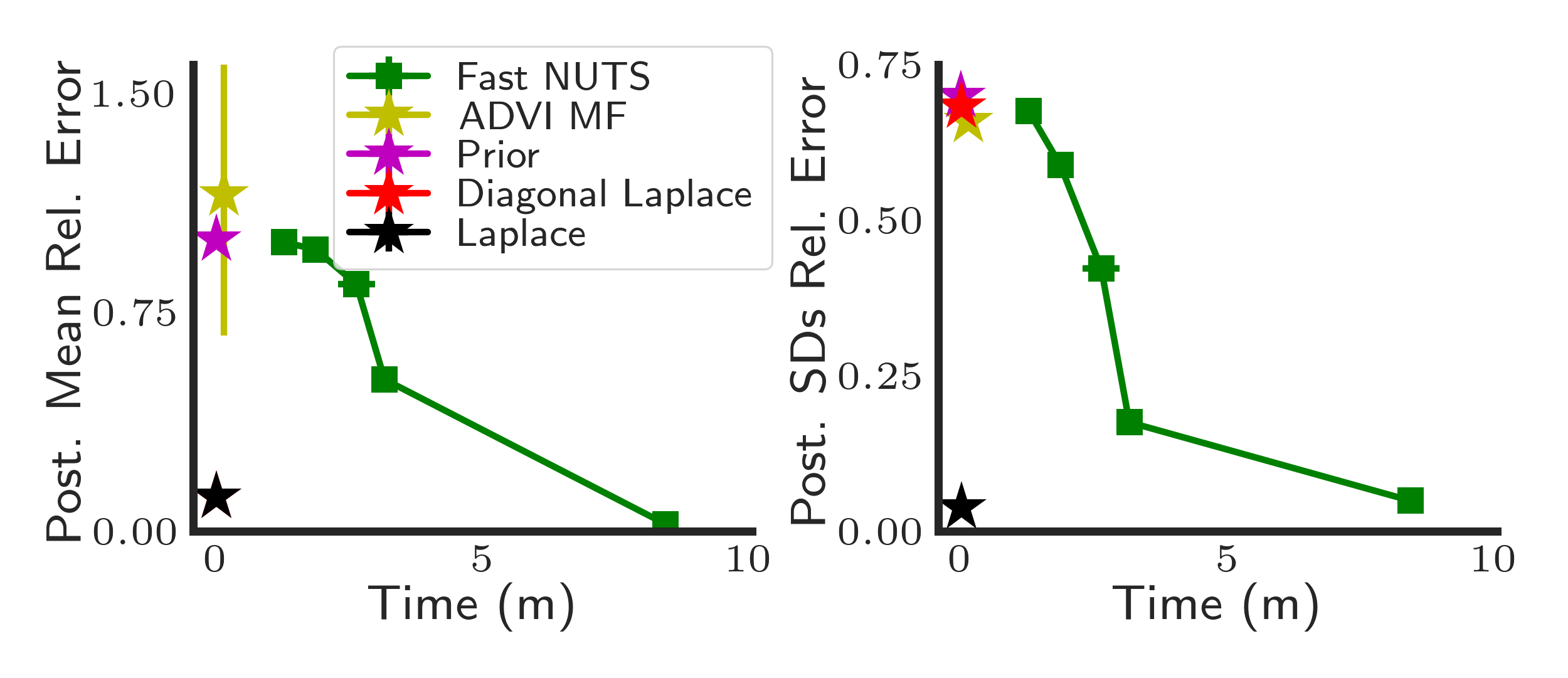}}
    {\caption{This figure is analogous to \Cref{fig:comp_trade_off_laplace}A but assesses LR-MCMC using NUTS in \texttt{Stan} rather than \methodnamelap. $D=250$.}
\label{fig:comp_trade_off_hmc}}
\end{figure}

\begin{figure}[!ht]
    \centering

    {\includegraphics[width=0.4\linewidth]{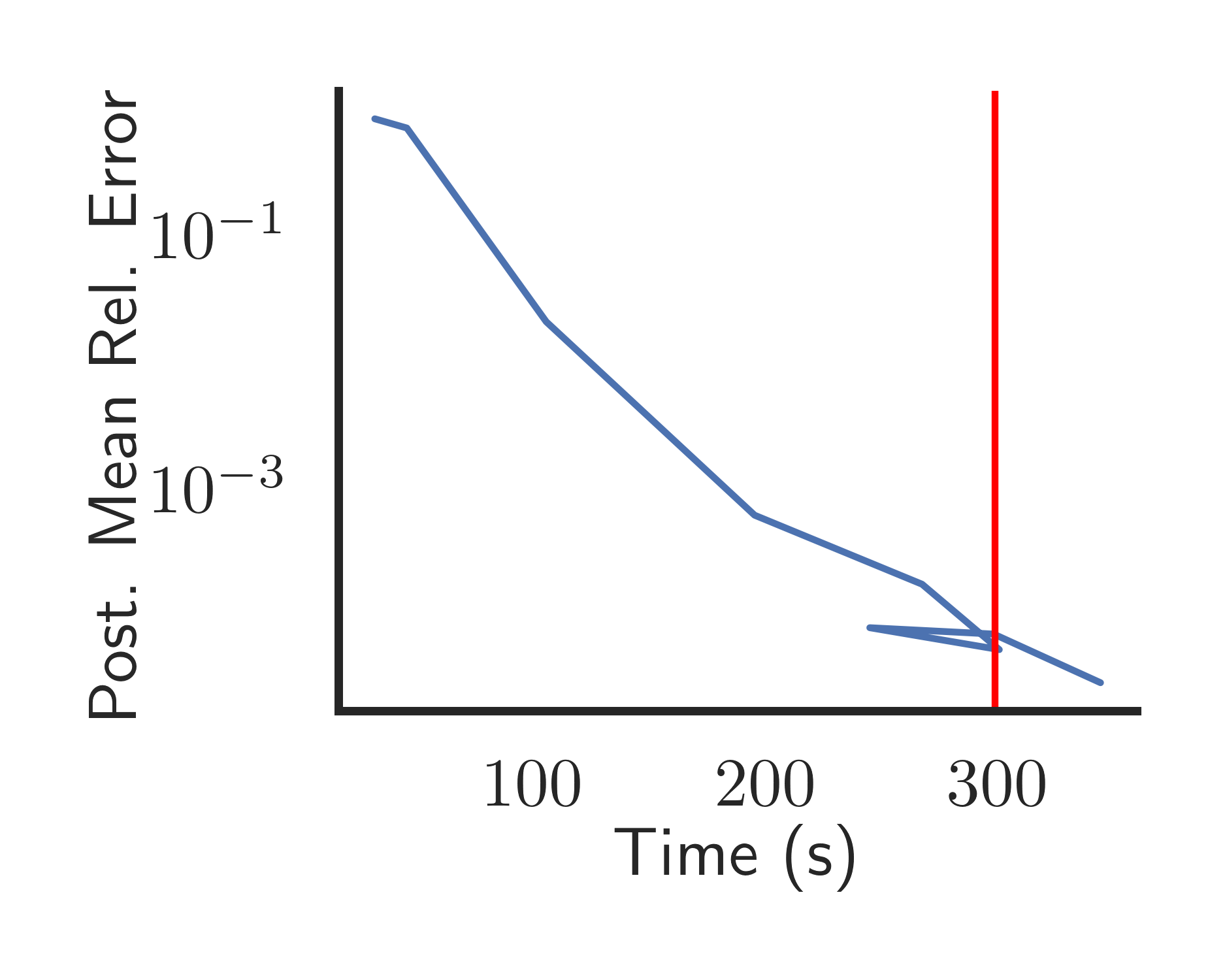}}{\includegraphics[width=0.4\linewidth]{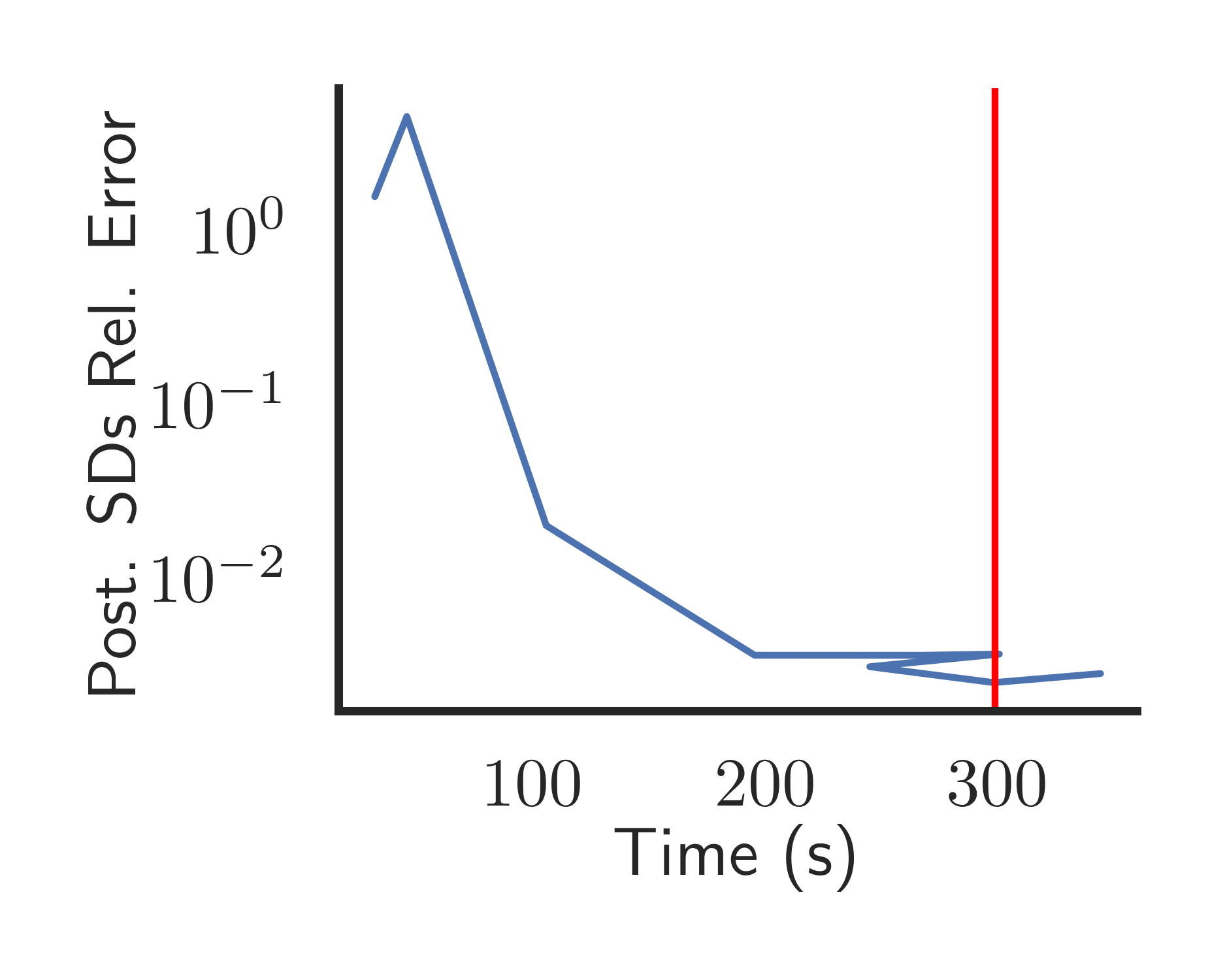}}
\abovecapvspace
    {\caption{Bayesian logistic regression with a regularized Horseshoe prior using NUTS in \texttt{Stan}. The red vertical line indicates the runtime of inference with \texttt{Stan} using the exact likelihood.}\label{fig:horseshoe}}
    \belowcapvspace
\end{figure}

\subsubsection{Horseshoe logistic regression experiment}
For the logistic regression experiment using a regularized horseshoe prior we used $N=1{,}000$ data points of dimension $D={200}$.
We used ten non-zero effects, each of size 10.  
Our implementation of the regularized horseshoe and inference in \texttt{Stan} closely followed M.\ Betancourt's  ``Bayes Sparse Regression'' case study.\footnote{\url{https://betanalpha.github.io/assets/case_studies/bayes_sparse_regression.html}}  We generated covariates as described in the previous section.

\subsection{Stan Model Code}\label{sec:stan_code}

First we show \texttt{Stan} code for Bayesian logistic regression. 
\begin{verbatim}
data {
 int<lower=1> N; // # of data
 int<lower=1> D; // # of covariates
 matrix[N, D] X; // Design matrix
 int<lower=0> y[N]; // labels
 real<lower=0> sigma;
}
parameters {
 vector[D] beta;
}
model {
 beta ~ normal(0, sigma);
 y ~ bernoulli_logit(X * beta);
}
\end{verbatim}

Second, we show \texttt{Stan} code for logistic regression with our low-rank approximation. 
\begin{verbatim}
data {
 int<lower=1> N; // # of data
 int<lower=1> D; // # of covariates
 int<lower=1> M; // Projected dimension
 matrix[D, M] U; // Projection matrix
 matrix[N, M] barX; // Projected design matrix
 int<lower=0> y[N]; // labels
 real<lower=0> sigma;
}
parameters {
 vector[D] beta;
}

transformed parameters {
 vector[M] bar_beta = U' * beta;
}
model {
 beta ~ normal(0, sigma);
 y ~ bernoulli_logit(barX * bar_beta);
}
\end{verbatim}


\section{Related Work on Scalable Bayesian Inference}\label{sec:related_work}
Developing scalable approximate Bayesian inference for models with many parameters (large $D$) and many data points (large $N$) has been active area of research for decades, and researchers have developed a large variety of methods applicable to GLMs.
Historically, Markov chain Monte Carlo (MCMC) methods based on the Metropolis-Hastings algorithm \cite{metropolis1953equation,hastings1970monte} have been dominant.
However MCMC is computationally expensive on large-scale problems in which both $D$ and $N$ are very large.
In particular, each likelihood evaluation requires $O(DN)$ time, due to the matrix vector product $X\beta$. Further, estimating posterior covariances uniformly well requires $O(\log D)$ samples \citep{cai2010optimal}. Therefore, the total cost of collecting those samples is $O(ND\log D)$ time in the case of perfect, independent Monte Carlo samples.
In practice, though, mixing times may also have unfavorable scaling with dimensionality and sample size; these issues can lead to even worse scaling in $N$ and $D$.
Several lines of research have explored the use of subsampling methods to reduce the dependence on $N$. 
But these methods either lose the asymptotic guarantees of exact MCMC or fail to provide faster inference in practice due to poor mixing behavior \citep{bardenet2017markov}.

Other work has pursued deterministic approximations to the Bayesian posterior.
Some of the most widely used of these approximations include (1) the Laplace approximation, which is a Gaussian approximation of the posterior defined locally at the posterior mode, (2) extensions of the Laplace approximation such as the integrated nested Laplace approximation (INLA) \citep{rue2009approximate}, and (3) variational Bayes; see, e.g., \citep[Chap.\ 10]{Bishop2006} and \citep{blei2017variational}.
However, these approaches also scale poorly with dimension in general.
The Laplace approximation requires computing and inverting the Hessian of the log posterior which demand $O(ND^2)$ and $O(D^3)$ time respectively, in order to compute approximate posterior means and variances.
In the $N \ll D$ setting, this cost can be reduced to $O(N^2 D)$ time (\Cref{sec:fast_inversions}). However, in large-$N$ settings of interest, the $O(N^2 D)$ cost can be prohibitive as well.
The cost of inference is further compounded when we give a fully Bayesian treatment to model hyperparameters as well as parameters; e.g., INLA requires this heavy computation for each nested approximation.
In the face of difficulties posed by high dimensionality, practitioners frequently turn to factorized (or ``mean-field'') approximations.
In the case of VB, the mean-field approach can yield biased approximations that underestimate uncertainty \cite{mackay2003information,turner2011two}. Likewise, factorized Laplace approximations, which approximate the Hessian with only its diagonal elements, similarly underestimate uncertainty (\Cref{sec:bad_marginals}).

Some more recent work has approached scalable approximate inference in generalized linear models with theoretical guarantees on quality in the large-$N$ regime by using likelihood approximations that are cheap to evaluate  \cite{Huggins2017a,campbell2017automated,pmlr-v80-campbell18a,huggins2016coresets}. But these methods fail to scale well to the large-$D$ case.

More closely related to the present work, \citet{geppert2017random} and \citet{lee2013bayesian} focus on conjugate Bayesian regression, respectively using random projections and principle component analysis to define low-rank descriptions of the design.
\citet{lee2013bayesian} restrict their consideration to the exactly low-rank case and primarily discuss the asymptotic consistency of the resulting posterior mean without discussing computational considerations.
\citet{spantini2015optimal} use conjugate Bayesian regression as stepping-off point to derive a point estimator for Bayesian inverse problems.
\citet{guhaniyogi2015bayesian} use random projections for Bayesian GLMs but focus on predictive performance rather than parameter estimation.
Outside the Bayesian context, \citet{zhang2014random}, \citet{wang2017sketching}, and many others have analyzed random projections for regression and classification using, for example, an M-estimation framework.

\section{Fast matrix inversions in the $N \ll D$ setting}\label{sec:fast_inversions}

In this section we focus on Gaussian conjugate linear regression with $N \ll D$. In this case, we can detail formulas for more efficient computation of the posterior mean and covariance.
We start from the standard expressions for the posterior mean $\mu_N$ and covariance $\Sigma_N$ when the prior is mean zero with covariance $\Sigma_{\beta}$; see \Cref{sec:proposal} and \Cref{sec:conj_low_rank} for further notation and setup of the model. These expressions are:
\[
\Sigma_N\inv &=\Sigma_\beta\inv + \tau X^\top X \label{eq:post-cov} \\
\mu_N &= \tau \Sigma_N X^\top Y. \label{eq:post-mean}
\]
Using these formulas naively in the $D \gg N$ setting is computationally expensive due to the $O(D^3)$ time cost of matrix inversion and $O(D^2)$ storage cost.

Using the Woodbury matrix identity, $(A\inv+UCV)\inv= A - AU(C\inv +VAU)\inv VA$, allows us to write $\Sigma_N = (\Sigma_\beta\inv + X^\top (\tau I_N) X)\inv$ as
\[
  \Sigma_N = \Sigma_\beta - \Sigma_\beta X^\top(\tau\inv I_N + X\Sigma_\beta X^\top)\inv X\Sigma_\beta. \label{eq:woodbury_cov}
\]
Computing $\Sigma_N$ via \Cref{eq:woodbury_cov} requires only $O(DN^2)$ cost for the matrix multiplications and an $O(N^3)$ cost for the matrix inversion. The posterior mean $\mu_N$ may then be computed in $O(ND)$ time by multiplying through by $X^\top Y$. These time costs can be significant reductions over the naive $O(D^3)$ cost when $N \ll D$.

\textbf{Fast inversions for the Laplace approximation to the GLM posterior}

We here show that the same approach described above may be used for the Laplace approximation in the context of Bayesian GLMs.  We say that we have a GLM likelihood if we can write
$$
	p( Y \mid \beta, X ) = \sum_{n=1}^N \phi(y_n, x_n^\top \beta)
$$
for some \emph{mapping function} $\phi: \R \times \R \rightarrow \R$. The Bayesian posterior then becomes
\[
    \log p(\beta \mid X, Y) = \log p(\beta) + \sum_{n=1}^N \phi(y_n, x_n^\top \beta) + Z,
\]
where $Z$ is a typically-intractable log normalizing constant.

Due to the analytic intractability of posterior inference in many common GLMs, approximations are necessary; the Laplace approximation is a particularly widely used approximation and takes the form
\[
    \bar p (\beta) =\mathcal{N}(\beta \mid \bar \mu, \bar \Sigma),
\]
where $\bar \mu := \argmax_\beta \log p (\beta \mid X, Y)$ and $\bar \Sigma := \left( - \nabla_\beta^2 \log p(\beta | X, Y)|_{\beta =\bar \mu} \right)\inv$.  However, as in the conjugate case, computing this matrix inverse naively can be expensive in the high-dimensional setting, and we are motivated to consider more computationally efficient routes to evaluate it.  In settings when $N \ll D$ and when we have a Gaussian prior $p(\beta)=\mathcal{N}(\beta \mid \mu_\beta, \Sigma_\beta)$, we may take an approach similar to our approach in the conjugate case. We first note
\[
    \nabla_\beta^2 \log p(\beta \mid X, Y)|_{\beta=\bar \mu} = -\Sigma_\beta\inv + X^\top \diag(\phipp(Y, X\bar \mu)) X,
\]
where $\phipp(Y,A)$ is a vector in $\R^N$ defined such that for any $n$ in $1,2, \dots, N$, $\phipp(Y, A)_n:=\frac{d^2}{da^2}\phi(y_i, a)|_{a=A_n}$.  Applying the same trick to this expression as before, we obtain
\[\label{eqn:woodbury_for_laplace}
  \bar \Sigma_N =\left( -\nabla_\beta^2 \log p(\beta \mid X, Y)|_{\beta=\bar \mu}\right)\inv = \Sigma_\beta - \Sigma_\beta X^\top \big(\diag[-\phipp(Y,X\bar \mu )]\inv + X\Sigma_\beta X^\top  \big)\inv X \Sigma_\beta,
\]
which again can yield computational gains.

It is worth noting however that this route is more computationally efficient only when the prior covariance matrix is structured in some way that allows for fast matrix-vector and matrix-matrix multiplications. This will be the case, for example, if $\Sigma_\beta$ is diagonal, block-diagonal, banded diagonal, or diagonal plus a low-rank matrix.

\section{Conjugate Gaussian regression with exactly low rank design}

\subsection{Derivation of \Cref{eqn:low_rank_blr}}\label{sec:low_rank_blr_proof}

Here we consider the setting of conjugate Bayesian linear regression, with $X$ exactly low rank and $\Sigma_\beta = \sigma_\beta^2 I_D$, as detailed in \Cref{sec:conj_low_rank}.
We now derive the expressions (\Cref{eqn:low_rank_blr}) for the mean and covariance of the Gaussian posterior for $\beta$ in this case. We suppose $X=V\mathrm{diag}(\lambda)U^\top $ for $U, V$ matrices of orthonormal rows and $\lambda$ a vector. The preceding equation for $X$ will capture low rank structure when $U \in \mathbb{R}^{D \times M}$ for some $M$ with $M \ll \min(D,N)$.

For the covariance, we start from \Cref{eq:post-cov}. Then we can rewrite $\Sigma_N$ as follows.
\begin{align*}
  \Sigma_N &= \big(\sigma_\beta^{-2}I_D + \tau X^\top X\big)\inv \\
      &= \big(\sigma_\beta^{-2}I_D + \tau U\diag(\lambda)V^\top V \diag(\lambda)U^\top \big)\inv \\
      &= \big(\sigma_\beta^{-2}I_D + U\diag(\tau \lambda \odot \lambda)U^\top \big)\inv \\
      & \textrm{where $\odot$ denotes component-wise multiplication, in this case across the components of the vector $\lambda$} \\
      &= \sigma_\beta^2 I- \sigma_\beta^2 U (\mathrm{diag}(\tau \lambda \odot \lambda)\inv + \sigma_\beta^2 I_M)\inv U^\top 
\sigma_\beta^2 \\
      & \textrm{by the Woodbury matrix identity and $U^\top U = I_M$} \\
    &= \sigma_\beta^2 I- \sigma_\beta^2 U
\mathrm{diag}\left\{ \left( \frac{1}{\tau \lambda \odot \lambda} +
\sigma_\beta^2 1_M \right)\inv \sigma_\beta^2 \right\} U^\top \\
      & \textrm{where division within the diag input is component-wise and $1_M$ is the all-ones vector of length $M$} \\
      &= \sigma_\beta^2 \left(
      		I_D - U \mathrm{diag}\left\{
      			\frac{\tau \lambda \odot \lambda }{\sigma_\beta^{-2} 1_M + \tau \lambda \odot \lambda }
		\right\} U^\top
		\right).
\end{align*}

Starting from \cref{eq:post-mean}, we can rewrite the posterior mean as follows.
\begin{align*}
  \mu_N &= \tau \Sigma_N X^\top Y\\
  	&= \tau \sigma_\beta^2 \left(
      		I_D - U \mathrm{diag}\left\{
      			\frac{\tau \lambda \odot \lambda }{\sigma_\beta^{-2} 1_M + \tau \lambda \odot \lambda }
		\right\} U^\top
		\right) U \diag(\lambda) V^\top Y
	\\
	& \textrm{from the derivation above and substituting for $X$} \\
	&= \tau \sigma_\beta^2 \left(
      		U - U \mathrm{diag}\left\{
      			\frac{\tau \lambda \odot \lambda }{\sigma_\beta^{-2} 1_M + \tau \lambda \odot \lambda }
		\right\}
		\right) \diag(\lambda) V^\top Y \\
	& \textrm{since $U^\top U = I_M$} \\
	&= \tau \sigma_\beta^2 U \left(
      		I_M - \mathrm{diag}\left\{
      			\frac{\tau \lambda \odot \lambda }{\sigma_\beta^{-2} 1_M + \tau \lambda \odot \lambda }
		\right\}
		\right) \diag(\lambda) V^\top Y \\
	&=  U \mathrm{diag}\left\{
      			\frac{\tau \lambda }{\sigma_\beta^{-2} 1_M + \tau \lambda \odot \lambda }
		\right\} V^\top Y.
\end{align*}

\section{Proofs and further results for conjugate Bayesian linear regression with low-rank data approximations}

\subsection{Proof of \Cref{thm:bayes_lin_reg_approx_quality} }\label{proof:bayes_lin_reg_approx_quality}
Recall that for conjugate Gaussian Bayesian linear regression, the exact posterior is $p(\beta \mid X, Y)=\mathcal{N}(\beta \mid \mu_N, \Sigma_N)$, where $\mu_{N}$
and $\Sigma_{N}$ are given in \cref{eq:post-mean,eq:post-cov}. 

Using an orthonormal projection $U$ yields a Gaussian approximate posterior
$\tilde p(\beta \mid X, Y)=\mathcal{N}(\beta \mid \tilde \mu_N, \tilde \Sigma_N)$. Recall from \Cref{sec:proposal} that we obtain this approximate posterior by replacing $X$ with $X U U^\top$. Thus, we can find $\tilde \mu_N$ and $\tilde \Sigma_N$ by consulting \Cref{eq:post-mean,eq:post-cov}:
\begin{align}\label{eqn:mean_and_cov_general}
\tilde \Sigma_N\inv &= \Sigma_\beta\inv + \tau UU^\top X^\top XUU^\top \\
\tilde\mu_N &= \tilde \tau \Sigma_N UU^\top X^\top Y.
\end{align}

\textbf{Upper bound on the posterior mean approximation error}

We will obtain our upper bound on the error of the approximate posterior mean relative to the exact posterior mean by upper bounding the norm of the difference between the gradient of the log posterior with respect to $\beta$ at the approximate posterior mean, $\tilde \mu_N$, and the exact posterior mean, $\mu_N$.
Together with the strong convexity of the negative log posterior, this bound will allow us to arrive at the desired upper bound on $\|\mu_N - \tilde \mu_N\|_2$.

First, we bound the norm of the gradient difference. To that end, the gradients of the exact log likelihood and the approximate log likelihood are given by
\begin{align*}
  \nabla_\beta \log p(Y \mid X, \beta) &= \nabla_\beta\left[ - \frac{\tau}{2}(X\beta - Y)^\top (X\beta - Y) \right] \\
                   &= -\tau (X^\top X\beta - X^\top Y)
\end{align*}
and
\begin{align*}
  \nabla_\beta \log \tilde p(Y \mid X, \beta) &= \nabla_\beta\left[ - \frac{\tau}{2}(XUU^\top \beta - Y)^\top (XUU^\top \beta - Y) \right] \\
                      &= -\tau (UU^\top X^\top XUU^\top \beta - UU^\top X^\top Y).
\end{align*}
We can thus upper bound the norm of the difference between the two log posteriors as follows.
\begin{align}
  \nonumber
  \lefteqn{ \left\| \nabla_\beta \log \tilde p(\beta \mid X, Y) - \nabla_\beta \log p(\beta \mid X, Y) \right\|_2 } \\
  \nonumber
   & = \left\| \nabla_\beta \log \tilde p(Y \mid X, \beta) - \nabla_\beta \log p(Y \mid X, \beta) \right\|_2 \\
   \nonumber
   &\textrm{since the prior is the same in both the exact and approximate model } \\
   \nonumber
   &\quad \textrm{and since the normalizing constant has no $\beta$ dependence} \\
   \nonumber
   &= \left\|
   	-\tau \left(
   			UU^\top X^\top XUU^\top \beta - UU^\top X^\top Y
		\right) + \tau\left(
			X^\top X\beta - X^\top Y
		\right)
	\right\|_2 \\
\nonumber
   &= \tau \left\|
   		\left(
			X^\top X - UU^\top X^\top XUU^\top
		\right)\beta + UU^\top X^\top Y - X^\top Y
	\right\|_2 \\
	\nonumber
   &= \tau \left\| \bar U\bar U^\top X^\top X\bar U\bar U^\top \beta - \bar U\bar U^\top X^\top Y \right\|_2 \\
   \nonumber
   & \textrm{where $\bar U$ (above) as well as $\bar{\lambda}$ and $\bar{V}$ (below) are defined in \Cref{sec:proposal}} \\
   \nonumber
   &= \tau \left\| \bar U \diag(\bar{\lambda} \odot \bar{\lambda}) \bar U^\top \beta - \bar U \diag(\bar{\lambda}) \bar V^\top Y \right\|_2 \\
   \nonumber
    &\le \tau \left( \left\| \bar U \diag(\bar{\lambda} \odot \bar{\lambda}) \bar U^\top \beta \|_2 + \| \bar U \diag(\bar{\lambda}) \bar V^\top Y \right\|_2 \right)\\
   \nonumber
   & \textrm{by the triangle inequality} \\
   \nonumber
    &= \tau \left( \left\| \diag(\bar{\lambda} \odot \bar{\lambda}) \bar U^\top \beta \|_2 + \| \diag(\bar{\lambda}) \bar V^\top Y \right\|_2 \right) \\
   \nonumber
   &\textrm{since $\| v \|_2^2 = v^\top v$ for a vector $v$ and $U^\top U = I_M$} \\
    \nonumber
    &\le \tau \left( \left\| \diag(\bar{\lambda} \odot \bar{\lambda})\right\|_{\mathrm{op}} \left\| \bar U^\top \beta \|_2 + \| \diag(\bar{\lambda}) \|_{\mathrm{op}} \| \bar V^\top Y \right\|_2 \right) \\
   \nonumber
   &\textrm{by definition of the operator norm in this space} \\
    \label{eq:grad-upper-bound}
    &= \tau \left( \bar\lambda_1^2 \| \bar U^\top \beta \|_2 + \bar\lambda_1 \| \bar V^\top Y\|_2 \right) \\
    \nonumber
\end{align}

Second, we need a result that will let us use the strong convexity of the negative log posterior.
We prove the following result in \Cref{sec:strong-convexity-proof}.
\begin{lemma}\label{lemm:strong-convexity}
    Let $f, g$ be twice differentiable functions mapping $\R^D \rightarrow \R$ and attaining minima at $\beta_f = \argmin_\beta f(\beta)$ and $\beta_g= \argmin_\beta g(\beta)$, respectively.  Additionally, assume that $f$ is $\alpha$--strongly convex for some $\alpha >0$ on the set $\{t \beta_f + (1-t)\beta_g | t \in [0, 1] \}$ and that $\| \nabla_\beta f(\beta_g) - \nabla_\beta g(\beta_g)\|_2= \|\nabla_\beta f(\beta_g)\|_2 \le c$. Then
\[
\| \beta_f - \beta_g \|_2 \le \frac{c}{\alpha}.
\]
\end{lemma}

To use the preceding result, we need a lower bound on the strong convexity constant of the negative log posterior; we now calculate such a bound.
We have that $\mu_N$ and $\tilde \mu_N$ are the maximum a posteriori values of $\beta$ under $p(\beta|X,Y,\alpha)$ and $\tilde p(\beta|X,Y,\alpha)$, respectively; equivalently they minimize the respective negative log of these distributions.
For a matrix $A$, let $\lambda_{\min}(A)$ denote its minimum eigenvalue. 
The Hessian of the negative log posterior with respect to $\beta$ is precisely $\Sigma_\beta\inv+ \tau X^\top X$ everywhere. So the negative log posterior is $\alpha$--strongly convex, where
\begin{equation}
\label{eq:strong_convexity_lower_bound}
\alpha = \lambda_{\min}(\Sigma_\beta\inv + \tau X^\top X) \ge \lambda_{\min}(\Sigma_\beta\inv) + \tau\lambda_{\min}(X^\top X) = \|\Sigma_\beta\|_2\inv +\tau \bar \lambda_{D-M}^2.
\end{equation}
In the first part of the final equality above, we use that the spectral norm of a matrix inverse is equal to the reciprocal of the minimum eigenvalue of the matrix.

Now we have an upper bound on the norm of the difference in gradients of the negative log posteriors (the same as for the log posteriors, in \Cref{eq:grad-upper-bound})
and a lower bound on the strong convexity constant from \Cref{eq:strong_convexity_lower_bound}.
So we can apply these together with \Cref{lemm:strong-convexity} to find
\begin{align*}
    \| \mu_N - \tilde \mu_N \|_2 
    &\le \frac{\tau \big(\bar\lambda_1^2 \| \bar U^\top \tilde \mu_N \|_2 + \bar\lambda_1 \|\bar V^\top Y\|_2\big)}{\alpha} \\
    &\textrm{by \Cref{lemm:strong-convexity} taking $\log p(\beta|X,Y)$ and $\log \tilde p(\beta | X, Y)$ }\\
    &\textrm{as $f$ and $g$ respectively, with $c$ given by \Cref{eq:grad-upper-bound}} \\
    &\le \frac{ \tau \big( \bar \lambda_1^2 \|\bar U^\top \tilde \mu_N\|_2 +  \bar \lambda_1 \|\bar V^\top  Y\|_2\big)}{\|\Sigma_\beta\|_2\inv + \tau \bar \lambda_{D-M}^2}\\
    &\textrm{by \Cref{eq:strong_convexity_lower_bound}} \\
    &= \frac{ \bar \lambda_1  \big( \bar \lambda_1 \|\bar U^\top \tilde \mu_N\|_2 +  \|\bar V^\top  Y\|_2\big)}{\|\tau \Sigma_\beta\|_2\inv + \bar \lambda_{D-M}^2}.
\end{align*}

Notably, in the common special case that $\Sigma_\beta$ is diagonal, as we saw in \Cref{sec:conj_low_rank}, $\tilde \mu_N$ will be in the span of $U$, and we will have that $\|\bar U^\top \tilde \mu_N\|_2=0$.

\textbf{Error in Posterior Precision}

The error in the precision matrices for the approximate and exact posteriors in linear regression are particularly straightforward since they do not depend on the responses, $Y$. In particular, we have
\begin{align}
  \Sigma_N\inv - \tilde\Sigma_N\inv &= (\Sigma_\beta\inv + \tau X^\top X) - (\Sigma_\beta\inv + \tau UU^\top X^\top XUU^\top ) \\
                 &= \tau X^\top X - \tau UU^\top X^\top XUU^\top \\
                 &= \tau \bar U \bar U^\top X^\top X\bar U\bar U^\top \\
                 &= \tau \bar U \diag(\bar\lambda \odot \bar\lambda) \bar U^\top.
\end{align}
Thus, since it is equal to the maximum eigenvalue, the spectral norm of the error in the precisions is precisely $\| \Sigma_N\inv - \tilde\Sigma_N\inv\|_2 = \tau \bar\lambda_1^2$.

\begin{figure}[!t]
\centering
 \includegraphics[width=.8\linewidth]{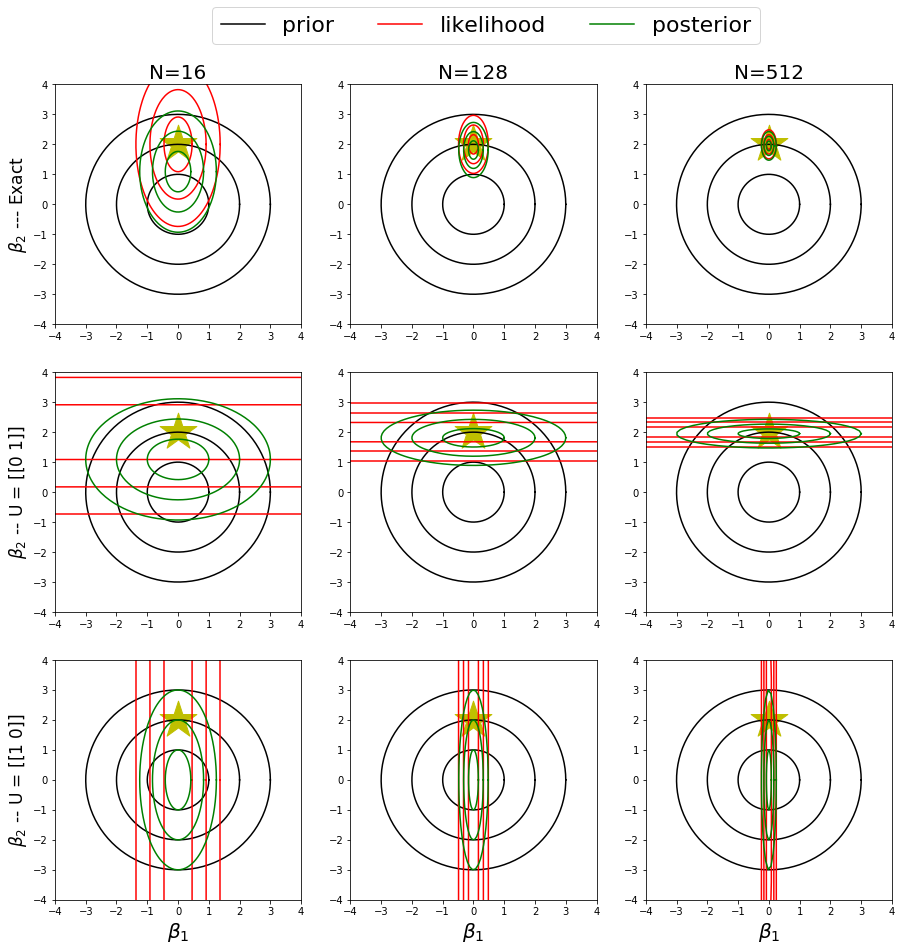}
  \caption{Example of posterior approximations with different projections (characterized by $U$) for increasing sample sizes.
    Each plot shows the contours of three densities: the prior, likelihood, and posterior (or approximations thereof).
    The top row shows the exact posterior.
    The middle row shows the approximations found by using the best rank-1 approximation to $X$.
    The bottom row shows the approximations found using the orthogonal rank-1 approximation.
    The star is at the parameter value used to generate simulated data for these plots. \label{fig:toy_example}}
\end{figure}

\subsection{Proof of \cref{lemm:strong-convexity}} \label{sec:strong-convexity-proof}

By the fundamental theorem of calculus, we may write
    $$\nabla_\beta  f(\beta) = \nabla_\beta f(\beta_g) +   \int_{t=0}^{1} (\beta - \beta_g)^\top \nabla_{\beta}^2 f(t \beta + (1-t) \beta_g )dt.$$

Considering the norm of $\nabla_\beta f(\beta)$ and applying the triangle inequality provides that for any $\beta$ in $\{t\beta_f + (1-t)\beta_g \mid t \in [0,1] \}$,
\begin{align}
    \| \nabla_\beta  f(\beta)\|_2 & \ge \left \|\int_{t=0}^{1} (\beta - \beta_g)^\top \nabla_{\beta}^2 f(t \beta + (1-t) \beta_g )dt\right\|_2 -  \|\nabla_\beta f(\beta_g) \|_2  \\
    & \ge \|\beta-\beta_g\|_2 \left\| \int_{t=0}^{1}  \nabla_{\beta}^2 f(t \beta + (1-t) \beta_g )dt\right\|_2 -  \|\nabla_\beta f(\beta_g) \|_2  \\
    & \ge \|\beta-\beta_g\|_2 \alpha -  \|\nabla_\beta f(\beta_g) \|_2.
\end{align}

We consider this bound at $\beta_f$. Recall we assume that $\|\nabla_\beta f(\beta_g) \|_2 \le c$. And $\|\nabla_\beta f(\beta_f)\|_2=0$ since $f$ is twice differentiable. Therefore, we have that $0 \ge \|\beta_f -\beta_g\|_2\alpha - c$, and the result follows.

%
%
%

\subsection{Proof of  \Cref{cor:not_consistent}}\label{sec:proof_not_consistent}



Our approach is to show that 
\[
\tilde \mu_N \stackrel{p}{\rightarrow} \Sigma_\beta U_* (U_*^\top \Sigma_\beta U_*)\inv U_*^\top \beta.  \label{eq:approx-mu-convergence}
\]
We then appeal to the following result, which we prove in \Cref{sec:min-norm-proof}:

\begin{lemma}\label{lemm:min-norm}
$\tilde \mu := \Sigma_\beta U (U^\top \Sigma_\beta U)\inv U^\top \beta$ is the vector of minimum $\Sigma_\beta\inv $-norm satisfying $U^\top \tilde \mu = U^\top \beta $.
\end{lemma}
Finally, for any closed $S\subset \R^D$, $\tilde \mu= \argmin_{v\in S} \|v\|_{\Sigma_\beta\inv } = \argmax_{v \in S} -\frac{1}{2} v^\top \Sigma_\beta\inv v=\argmax_{v \in S} \mathcal{N}(0, \Sigma_\beta)$. Therefore, the $\tilde \mu$ in \Cref{lemm:min-norm} is the maximum a priori vector satisfying the constraint in \Cref{lemm:min-norm}.

We first turn to proving \cref{eq:approx-mu-convergence}. 
Let $U_N\diag( \lambda^{(N)})V_N^\top$ denote the $M$-truncated SVD of the design matrix consisting of $N$ samples $X=(x_1, x_2, \dots, x_N)$ where $x_i \overset{\mathrm{i.i.d.}}{\sim} p_*$.   When the low rank approximation is defined by this SVD, from \Cref{eqn:mean_and_cov_general} we have that $\tilde \mu_N = \tau \tilde \Sigma_N U_NU_N^\top X^\top Y$. Noting that $Y=X\beta + \frac{1}{\tau}\epsilon$ for some $\epsilon \in \R^N$ with $\epsilon_i \distiid \mathcal{N}(0, 1)$, we may expand this out and write:
\begin{align}
    \nonumber
  \tilde \mu_N &= \tau (\Sigma_\beta\inv + U_NU_N^\top X^\top \tau X U_NU_N^\top )\inv U_NU_N^\top X^\top (X\beta + \frac{1}{\tau}\epsilon) \\
    \nonumber
    &= \tau \left\{ \Sigma_\beta\inv + U_N\left[\tau \diag(\lambda^{(N)}\odot \lambda^{(N)})\right] U_N^\top \right\}\inv U_N\diag(\lambda^{(N)}) V_N^\top \left[ V_N\diag(\lambda^{(N)}) U_N^\top \beta + \frac{1}{\tau}\epsilon\right] \\
    \nonumber
    &= \left\{ \Sigma_\beta\inv + U_N\left[\tau \diag(\lambda^{(N)}\odot \lambda^{(N)})\right]U_N^\top \right\}\inv U_N\left[\tau \diag(\lambda^{(N)}\odot \lambda^{(N)})\right] \left[U_N^\top \beta + \diag(\lambda^{(N)})\inv V_N^\top \frac{1}{\tau}\epsilon\right] \\
    \nonumber
    &= \Sigma_\beta U_N \left[U_N^\top \Sigma_\beta U_N +\tau\inv \diag(\lambda^{(N)})^{-2} \right]\inv \left[U_N^\top \beta + \diag(\lambda^{(N)})\inv V_N^\top \frac{1}{\tau}\epsilon\right] \\
    \nonumber
  &\convP \Sigma_\beta U_*(U_*^\top \Sigma_\beta U_*)\inv U_*^\top \beta,
\end{align}
where in the fourth line we use the matrix identity, $(R\inv  + W^\top QW)\inv W^\top Q = RW^\top (WRW^\top + Q\inv )\inv $ \citep{petersen2008matrix}. 
Convergence in probability in the last line follows since $\diag({\lambda^{(N)}}^{-2}) \convP 0$ \citep{vershynin2012close} and $U_N \convP U$.

\subsection{Proof of \cref{lemm:min-norm}}\label{sec:min-norm-proof}

 We show that $\beta_*=\Sigma_\beta U(U^\top \Sigma_\beta U)\inv U^\top \beta$ is the vector of minimum norm satisfying the above constraints in the Hilbert space $\R^D$ with inner product $\langle v_1, v_2 \rangle = v_1^\top \Sigma_\beta\inv v_2 $ for vectors $v_1, v_2 \in \R^D$.

Define $\beta_*$ as
\begin{align}\label{eqn:min_norm_problem}
  \beta_* = \argmin_{v\in \R^D} \| v\|_{\Sigma_\beta\inv } \,\,\, \mathrm{subject \, to}\,\,\, U^\top v = U^\top \beta
\end{align}
First note that the condition $U^\top \beta_*=U^\top \beta$ may be expressed as a set the $M$ linear constraints
\begin{align}\label{eqn:lin_constraints}
  \langle \Sigma_\beta  U[:, i], \beta_* \rangle = U[:, i]^\top \beta
\end{align}
  for $i=1,2,\dots, M$. We thereby see that the constraint restricts $\beta_*$ to the linear variety $\beta + \big[ \big\{ \Sigma_\beta U[:, i] \big\}_{i=1}^M \big]^\perp$, where $\big[A\big]$ denotes the subspace generated by the vectors of the set $A$ and $\big[ A\big]^\perp$ denotes the set of all vectors orthogonal to $\big[A\big]$ (i.e. the orthogonal complement of $\big[ A\big]$). 
  By the projection theorem \citep{luenberger1997optimization}, $\beta_*$ is orthogonal to $\big[ \big\{ \Sigma_\beta U[:, i] \big\}_{i=1}^M \big]^\perp$, or $\beta_* \in \big[ \big\{ \Sigma_\beta U[:, i] \big\}_{i=1}^M \big]^{\perp \perp}=\big[ \big\{ \Sigma_\beta U[:, i] \big\}_{i=1}^M\big]$. We can therefore write $\beta_*$ as a linear combination of the vectors $\big\{ \Sigma_\beta U[:, i] \big\}_{i=1}^M$; that is, for some $c$ in $\R^M$
\begin{align}\label{eqn:solution_form}
  \beta_* = \Sigma_\beta Uc.
\end{align}
  Our constraints in \Cref{eqn:lin_constraints} then demand that $\langle \Sigma_\beta U[:, i], \Sigma_\beta Uc \rangle = U[:, i]^\top \beta$ for each $i$, or equivalently that $U^\top \Sigma_\beta \Sigma_\beta\inv \Sigma_\beta U c = U^\top \beta$. This implies that $c=(U^\top \Sigma_\beta U)\inv U^\top \beta$. Plugging this into \Cref{eqn:solution_form} yields $\beta_* = \Sigma_\beta U (U^\top \Sigma_\beta U)\inv U^\top \beta$, as desired.

\subsection{Proof of \Cref{cor:conservative}}
Recall that we wish to show that, for conjugate Bayesian regression, under $\tilde p$ the uncertainty (i.e., posterior variance) for any linear combination of parameters, $\mathrm{Var}_{\tilde p}[v^\top \beta]$, is no smaller than the exact posterior variance.
First, we note that this statement is formally equivalent to stating that $v^\top \tilde \Sigma_N v \ge v^\top \Sigma_N v$, or that $E:=\tilde \Sigma_N-\Sigma_N \succeq 0$ (where $\succeq$ denotes positive definiteness).
By \Cref{thm:bayes_lin_reg_approx_quality}, $\Sigma_N\inv-\tilde \Sigma_N\inv =\bar U \mathrm{diag}({\bar \lambda}^2) \bar U^\top \succeq 0$.
Since this implies that the inverse of the difference of these matrices is positive definite, we can then see that $(\Sigma_N\inv - \tilde \Sigma_N\inv )\inv =\tilde \Sigma_N(\tilde \Sigma_N - \Sigma_N)\inv \Sigma_N\succeq 0$.
Because, as valid covariance matrices, $\Sigma_N$ and $\tilde \Sigma_N$ are both positive definite, and because inverses and product of positive definite matrices are positive definite, this implies that $\tilde \Sigma_N\inv \tilde \Sigma_N(\tilde \Sigma_N - \Sigma_N)\inv \Sigma_N\Sigma_N\inv =(\tilde \Sigma_N - \Sigma_N)\inv \succeq 0$.
Finally, this implies that $\tilde \Sigma_N - \Sigma_N \succeq 0$ as desired.

\subsection{Information loss due the \methodname approximation}\label{sec:information_gain}
We see similar behavior to that demonstrated in \Cref{cor:conservative} in the following corollary, which shows that our approximate posterior never has lower entropy than the exact posterior.
Concretely, we look at the reduction of entropy in the approximate posterior relative to the exact posterior \citep{mackay2003information}, where entropy is defined as:
$$H\big[ p(\beta)\big]:= \E_p[-\log_2 p(\beta)]$$
\begin{corollary}\label{cor:information_gain}
The entropy $H\big[ \tilde p(\beta|X, Y)\big]$ is no less than $H\big[p(\beta|X, Y)\big]$. Furthermore, when using an isotropic Gaussian prior $\Sigma_\beta=\sigma_\beta^2 I$, the information loss relative to the exact posterior (in nats) is upper bounded as $H\big[\tilde p(\beta|X, Y) \big] - H\big[ p(\beta|X, Y) \big] \le \frac{\tau\sigma_\beta^2}{2} \sum_{i=1}^{D-M} \bar \lambda_{i}^2$.
\end{corollary}

This result formalizes the intuition that the \methodname approximation reduces the information about the parameter that we are able to extract from the data.
Additionally, the upper bound tells us that when $U$ is obtained via an $M$ truncated SVD, at most $\tau\sigma_\beta^2 \bar \lambda_{1}^2/2$ additional nats of information would have been provided by using the $M+1$-truncated SVD.

\begin{proof}
The entropy of the exact and approximate posteriors are given as:
\(
    H(p) = -\frac{1}{2}\log | 2\pi e \Sigma_N\inv | = -\frac{1}{2}\big[D \log 2\pi e + \sum_{i=1}^D\log (\sigma_\beta^{-2}+\tau \lambda_{i}^2)\big]
\)
and
\(
    H(\tilde p) = -\frac{1}{2}\log | 2\pi e \tilde \Sigma_N\inv | =-\frac{1}{2}\big[D \log 2\pi e + \sum_{i=1}^M \log (\sigma_\beta^{-2}+\tau \lambda_{i}^2) - \sum_{i=M+1}^D \log \sigma_\beta^{-2}\big].
\)
Therefore, we conclude that 
\(
    H\big[\tilde p(\beta|X) \big] - H\big[ p(\beta|X) \big] &=  - \frac{1}{2}\sum_{i=1}^{D-M} \log \sigma_\beta^{-2} + \frac{1}{2}\sum_{i=1}^{D-M}\log (\sigma_\beta^{-2}+\tau\bar \lambda_{i}^2)\\
    &=\frac{1}{2}\sum_{i=1}^{D-M} \log \frac{\sigma_\beta^{-2} + \tau\bar\lambda_i^2}{\sigma_\beta^{-2}}\\
  &=\frac{1}{2}\sum_{i=1}^{D-M} \log (1+\frac{\tau}{\sigma_\beta^{-2}}\bar \lambda_{i}^2 )\\
    &\le \frac{1}{2}\sum_{i=1}^{D-M} \frac{\tau}{\sigma_\beta^{-2}} \bar \lambda_{i}^2 
    = \frac{\tau \sigma_\beta^2}{2} \sum_{i=1}^{D-M} \bar \lambda_{i}^2 .
\)
That $H\big[\tilde p(\beta|X) \big] - H\big[ p(\beta|X) \big]>0$ follows from the monotonicity of $\log$, that $\log (1)=0$, and that ${\tau}\sigma_\beta^{2} \bar \lambda_i^2 >0$ 
for $i=1,\dots,D-M$. 
\end{proof}

\section{Proofs and further results for \methodnamelap in non-conjugate models}

In the main text we introduced LR-Laplace as a method which takes advantage of low-rank approximations to provide computational gains when computing a Laplace approximation to the Bayesian posterior.  In what follows we verify the theoretical justifications for this approach.
\Cref{sec:fl_complexity_proof} provides a derivation of \Cref{alg:fast_laplace} and demonstrates the time complexities of each step, serving as a proof of \Cref{thm:laplace_time}.
The remainder of the section is devoted to the proofs and discussion of the theoretical properties of LR-Laplace.

\subsection{Proof of \Cref{thm:laplace_time}}\label{sec:fl_complexity_proof}

\begin{proof}[{Proof of  \Cref{thm:laplace_time}}]
The \methodnamelap approximation is defined by mean and covariance parameters, $\hat \mu$ and $\hat \Sigma$.
We prove \Cref{thm:laplace_time} in two parts.
First, we show that $\hat \mu$ and $\hat \Sigma$ do in fact define the Laplace approximation of $\tilde p(\beta|X, Y)$, i.e. the construction of $\hat \mu$ in \Cref{line:hat_mu} satisfies $\hat \mu = \argmax_\beta \tilde p(\beta | X, Y)$ and that $\hat \Sigma= \big(-\nabla_\beta^2 \log \tilde p(\beta|X, Y) |_{\beta=\hat\mu}\big)\inv$.
Second, we show that each step of \Cref{alg:fast_laplace} may be computed in $O(NDM)$ time with $O(DM+NM)$ storage.

\textbf{Correctness of $\hat \mu$ and $\hat \Sigma$}

In \Cref{line:gamma}, the definition of $\gamma_*$ implies that $\gamma_*=\argmax_{\gamma\in \R^M} \tilde p_{U^\top \beta|X,Y}(\gamma|X,Y)$ since 
\(
  \log \tilde p_{U^\top \beta|X,Y}(\gamma|X,Y) &= \log \tilde p_{U^\top \beta}(\gamma)+ \log \tilde p_{Y|X, U^\top \beta}(Y|X, \gamma)  + C\\
  &= \log p_{U^\top \beta}(\gamma)+ \log p_{Y|X, \beta}(Y|X, U\gamma) +C \\
  &= \log \mathcal{N}(\gamma|U^\top \mu_\beta, U^\top \Sigma_\beta U)+ \sum_{i=1}^N \log p_{y|x,\beta}(y_i|x_i, U \gamma) + C\\
  &= -\frac{1}{2}\gamma^\top U^\top \Sigma_\beta U\gamma + \sum_{i=1}^N \phi(y_i,x_i^\top U \gamma) + C^\prime,
\)
  where line 1 uses Bayes' rule, line 2 uses the definition of $\tilde p$ in \Cref{eqn:approximation}, line 3 uses the normality the prior, and the assumed conditional independence of the responses given $\beta$, and line 4 follows from the definition of $\phi(\cdot,\cdot)$ and the assumption that $\mu_\beta={0}$.
$C$ and $C^\prime$ are constants which do not depend on $\gamma$.
This together with the following result (proved in \Cref{sec:full-dim-beta-laplace-proof}) implies that as defined in \Cref{line:hat_mu} of \Cref{alg:fast_laplace}, $\hat \mu=\argmax_\beta \tilde p(\beta|X, Y)$.
  \begin{lemma}\label{lemm:full-dim-beta-laplace}
  Suppose a Gaussian prior $p(\beta)=\mathcal{N}(\mu_\beta, \Sigma_\beta)$, and let $\gamma_* := \argmax_{\gamma \in \R^M} \log \tilde p_{U^\top \beta|X,Y}(\gamma | X, Y)$. Then $\hat \mu := \argmax_{\beta \in \R^D} \log \tilde p(\beta | X, Y)$ may be written as $\hat \mu = U\gamma_* + \bar U \bar U^\top \Sigma_\beta U(U^\top\Sigma_\beta U)^{-1} \gamma_*$.
\end{lemma}

We now show that as defined in \Cref{line:hat_sigma} of \Cref{alg:fast_laplace}, $\hat \Sigma$ is inverse of the Hessian of the negative log posterior, $H$. We see this by writing
\(
  H :&= \nabla_\beta^2 -\log \tilde p(\beta|X, Y)|_{\beta = \hat \mu} \\
    &= \nabla_\beta^2 -\log \mathcal{N}(\beta|\mu_\beta, \Sigma_\beta)|_{\beta = \hat \mu} + \nabla_\beta^2 \sum_{i=1}^N -\phi(y_i, x_i^\top UU^\top \beta)|_{\beta = \hat \mu} \\
    &= \Sigma_\beta\inv+ \sum_{i=1}^N -\phi^{\prime\prime}(y_i, x_i^\top UU^\top \hat \mu) x_iUU^\top x_i^\top \\
    &= \Sigma_\beta\inv+ UU^\top X^\top \mathrm{diag}\big(-\phipp(Y, XUU^\top \hat \mu) \big)X U U^\top ,
\)

where $\phipp$ is the second derivative of $\phi$.
The Woodbury matrix lemma then provides that we may compute $\hat \Sigma_N:=H\inv $ as 
\(
    \hat \Sigma_N = \Sigma_\beta - \Sigma_\beta U\left(U^\top \Sigma_\beta U-\left\{ U^\top X^\top \mathrm{diag}\left[\phipp(Y, XUU^\top \hat \mu) \right] X U \right\}\inv \right)\inv U^\top \Sigma_\beta ,
\)
    which we have written as $\hat \Sigma := \Sigma_\beta - \Sigma_\beta U W U^\top \Sigma_\beta$ in \Cref{line:hat_sigma} with $W\inv=U^\top \Sigma_\beta U - \left\{ U^\top X^\top \diag\left[ \phipp(Y, XUU^\top \hat \mu)\right]XU\right\}\inv$.

\textbf{Time complexity of \Cref{alg:fast_laplace}}

We now prove the asserted time and memory complexities for each line of \Cref{alg:fast_laplace}.

\Cref{alg:fast_laplace} begins with the computation of the $M$-truncated SVD of $X^\top \approx U \diag(\mathbf{\lambda})V$.
As discussed in \Cref{sec:conj_low_rank}, $U$ may be found in $O(ND\log M)$ time.
At the end of this step we must store the projected data $XU \in \R^{N,M}$ and the left singular vectors, $U\in \R^{D,M}$.
Which demands $O(NM+DM)$ memory, and the matrix multiply for $XU$ requires $O(NDM)$ time and is the bottleneck step of the algorithm. The matrix $V$ need not be explicitly computed or stored.

The next stage of the algorithm is solving for $\hat \mu = \argmax_\beta \log \tilde p(\beta | X, Y)$.
This is done in two stages: in \Cref{line:gamma} find $\gamma_* = \argmax_{\gamma \in \R^M} \log \tilde p_{U^\top \beta|X,Y}(\gamma | X, Y)$ as the solution to a convex optimization problem, and in \Cref{line:hat_mu} find $\hat \mu$ as $\hat \mu = U \gamma_* + \bar U \bar U^\top \Sigma_\beta U(U^\top\Sigma_\beta U)^{-1} \gamma_*$.
Beginning with \Cref{line:gamma}, we note that the function $\log \tilde p(U^\top \beta|X,Y) = \log p(\beta) + \log \tilde p(Y|X,\beta)+c\stackrel{c}{=} \log \mathcal{N}(U^\top \beta|U^\top \mu_\beta, U^\top \Sigma_\beta U) + \sum_{i=1}^N \log p(y_i|x_i^\top UU^\top \beta)$ is a finite sum of functions concave in $\beta$ and therefore also in $U^\top \beta$.
$\gamma_*$ may therefore be solved to a fixed precision in $O(NM)$ time under the assumptions of our theorem using stochastic optimization algorithms such as stochastic average gradient \cite{schmidt2017minimizing}.
In our experiments we use more standard batch convex optimization algorithm (L-BFGS-B \cite{zhu1997algorithm}) which takes at most $O(N^2M)$ time.
This latter upper bound on complexity may be seen from observing each gradient evaluation takes $O(NM)$ time
(the cost for the likelihood evaluation, since computing the log prior and its gradient is $O(M^2)$ after computing $U^\top \Sigma_\beta U$ once, which takes $O(DM^2)$ time by assumption)
and the number of iterations required can grow up to linearly in the maximum eigenvalue of Hessian, which in turn grows linearly in $N$ \cite{boyd2004convex}.

  The second step is computing $\hat \mu = U \gamma_* + \bar U \bar U^\top \Sigma_\beta U(U^\top\Sigma_\beta U)^{-1} \gamma_*$. Given $\gamma_*$, this may be computed in $O(DM)$ time, which one may see by noting that $\bar U \bar U^\top $ (which we never explicitly compute) may be written as $\bar U \bar U^\top =(I-UU^\top )$, and finding $\hat \mu$ as $\hat \mu = U \gamma_* + \Sigma_\beta U(U^\top\Sigma_\beta U)^{-1} \gamma_*- UU^\top \Sigma_\beta U(U^\top\Sigma_\beta U)^{-1} \gamma_*$. By assumption, the structure of $\Sigma_\beta$ allows us to compute $U^\top \Sigma_\beta U$ in $O(DM^2)$ time and matrix vector products with $\Sigma_\beta$ in $O(D)$ time.

  We now turn to the third stage of the algorithm, solving for the posterior covariance $\hat \Sigma$, which is represented as an expression of $U$, $\Sigma_\beta$ and $W$, defined in \Cref{line:W}.
Computing $W$ requires $O(DM)$ and $O(NM^2)$ matrix multiplications (since we have precomputed $XU$), and two $O(M^3)$ matrix inversions which comes to $O(NM^2+DM)$ time.
The memory complexity of this step is $O(NM)$ since it involves handling $XU$.
Once $W$ has been computed we may use the representation $\hat \Sigma=\Sigma_\beta - \Sigma_\beta U W U^\top \Sigma_\beta$ as presented in \Cref{line:hat_sigma}.
This representation does not entail performing any additional computation (which is why we have written $O(0)$), but as this expression includes $U$, storing $\hat \Sigma$ requires $O(DM)$ memory.

  Lastly, we may immediately see that computing posterior variances and covariances takes only $O(M^2)$ time as it involves only indexing into $\Sigma_\beta$ and $U$ and $O(M^2)$ matrix-vector multiplies.
\end{proof}

\subsection{Proof of \Cref{lemm:full-dim-beta-laplace}} \label{sec:full-dim-beta-laplace-proof}

  We prove the lemma by constructing a rotation of the parameter space by the matrix of singular vectors $[U, \bar U]$, in which we have the prior
$$p\bigg(\begin{bmatrix}U^\top \beta \\ \bar U^\top \beta \end{bmatrix}\bigg)=\mathcal{N}\Big(
  \begin{bmatrix}U^\top \beta \\ \bar U^\top \beta \end{bmatrix} \Big| 
  \begin{bmatrix}U^\top \mu_\beta \\ \bar U^\top \mu_\beta \end{bmatrix}, 
  \begin{bmatrix}U^\top \Sigma_\beta U, \,\, U^\top \Sigma_\beta \bar U \\ 
  \bar U^\top \Sigma_\beta U ,\,\, \bar U^\top \Sigma_\beta \bar U \end{bmatrix}
\Big).
$$
We have that
\(
  \hat \mu :&= \argmax_{\beta \in \R^D} \log \tilde p(\beta | X, Y)\\
  &= [U \, \bar U] \argmax_{U^\top \beta \in \R^M, \bar U^\top \beta \in \R^{D-M}} \log \tilde p\big(\begin{bmatrix}U^\top \beta \\ \bar U^\top \beta \end{bmatrix} | X, Y)\\
    &= U\argmax_{U^\top \beta \in \R^M}\big( \log \tilde p(U^\top \beta | X, Y) + \bar U\argmax_{\bar U^\top \beta \in \R^{D-M}} \log \tilde p(\bar U^\top \beta | U^\top \beta, X,Y) \big)\\
    &= U\argmax_{U^\top \beta \in \R^M} \log \tilde p(U^\top \beta | X, Y) + \\
  &\phantom{=~} \bar U\argmax_{\bar U^\top \beta \in \R^{D-M}} \log \mathcal{N}\big(\bar U^\top \beta | \bar U^\top \Sigma_\beta U (U^\top \Sigma_\beta U)^{-1}U^\top \beta, \bar U \Sigma_\beta \bar U - \bar U \Sigma_\beta U ( U^\top \Sigma_\beta U)U^\top \Sigma_\beta \bar U \big)\\
    &= U\gamma_*+ \bar U\bar U^\top \Sigma_\beta U (U^\top \Sigma_\beta U)^{-1}\gamma_* .
\)
In the second line we simply move to the rotated parameter space. In the third line, we use the chain rule of probability to separate out two terms.
To produce the fourth line, we note that since $\tilde p(Y|X, \beta) = p(Y|XUU^\top \beta)=\tilde p(Y|X, U^\top \beta)$, that $Y$ and $\bar U^\top \beta$ are conditionally independent given $U^\top \beta$.
We next note that though $\argmax_{\bar U^\top \beta \in \R^{D-M}} \log p(\bar U^\top \beta | U^\top \beta)$ depends on $U^\top \beta$, $\max_{\bar U^\top \beta \in \R^{D-M}} \log p(\bar U^\top \beta | U^\top \beta)$ does not depend $U^\top \beta$. This allows us to use the definition of $\gamma_*$ to arrive at the fifth line, as desired.

    In the special case that $\Sigma_\beta$ is diagonal, this expression reduces to $U\gamma_*$.  This can be seen by recognizing that $\bar U^\top \Sigma_\beta U$ is then $\diag(\mathbf{0})$.


\subsection{Proof of \Cref{thm:glm_mean_bound}}
Our approach to proving \Cref{thm:glm_mean_bound} follows a similar approach to that taken to prove \Cref{thm:bayes_lin_reg_approx_quality}.  In particular, we begin by upper bounding the norm of the error of the gradients at the approximate MAP.  Noting that the strong log concavity of the exact posterior, which having been assumed to hold globally, must then also hold on $\{t\hat\mu + (1-t)\bar \mu |  t \in [0,1]\}$, we obtain an upper-bound on $\| \hat \mu -\bar \mu\|_2$ by again applying \Cref{lemm:strong-convexity}.

To begin, we first recall that the exact and \methodname posteriors may be written as

$$\log p(\beta|X,Y) = \log p(\beta) + \sum_{n=1}^N \phi(y_n|x_n^\top \beta) -\log Z$$ \ and \ $$\log \tilde p(\beta|X,Y) = \log p(\beta) + \sum_{n=1}^N \phi(y_n,x_n^\top UU^\top \beta) -\log \tilde Z$$
where $\phi(\cdot, \cdot)$ is such that $\phi(y, a)=\log p(y|x^\top \beta= a)$, and  $Z$ and $\tilde Z$ are the normalizing constants of the exact and approximate posteriors.  As a result, the gradients of these log densities are given as
$$\nabla_\beta \log p(\beta|X,Y) = \nabla_\beta \log p(\beta) +  X^\top \phip(Y,X\beta)$$
and
$$\nabla_\beta \log \tilde p(\beta|X,Y) = \nabla_\beta \log p(\beta) +  UU^\top X^\top \phip(Y,XUU^\top \beta) ,$$
where $\phip(Y, X\beta)\in \R^N$ is such that for each $n \in [N]$, $\phip(Y, X\beta)_n= \frac{d}{da}\phi(y_n, a)|_{a=x_n^\top \beta}$.

And the difference in the gradients is
\begin{align}
    \nabla_\beta \log p(\beta|X,Y) - \nabla_\beta \log \tilde p(\beta|X,Y) &=  X^\top \phip(Y,X\beta)- UU^\top X^\top \phip(Y,XUU^\top \beta) .
\end{align}

Appealing to Taylor's theorem, we may write for any $\beta$ that
$$\phi^\prime(y_n, x_n^\top UU^\top \beta)=\phi^\prime(y_n, x_n^\top \beta) + (x_n^\top UU^\top \beta - x_n^\top  \beta) \phi^{\prime\prime}(y_n, a_n)$$
    for some $a_n \in [x_n^\top UU^\top \beta, x_n^\top  \beta]$, where $\phi^{\prime\prime}(y, a):= \frac{d^2}{da^2}\phi(y,a)$.

Using this and introducing vectorized notation for $\phi^{\prime \prime}$ to match that used for $\phip$, we may rewrite the difference in the gradients as
\begin{align*}
   \lefteqn{\nabla_\beta \log p(\beta|X,Y) - \nabla_\beta \log \tilde p(\beta|X,Y)} \\
    &=  X^\top \phip(Y,X\beta)- UU^\top X^\top \phip(Y,X\beta)-UU^\top X^\top \big[(X UU^\top\beta - X^\top\beta)\circ \phipp(Y,A)\big]\\
    &=\bar U\bar U^\top X^\top \phip(Y,X\beta)+ UU^\top X^\top \big[(X \bar U\bar U^\top\beta)\circ \phipp(Y,A)\big] ,
\end{align*}
where $A\in \R^N$ is such that for each $n\in [N]$, $A_n \in [x_n^\top UU^\top \beta, x_n^\top \beta]$, and $\circ$ denotes element-wise scalar multiplication.

We can use this to derive an upper bound on the norm of the difference of the gradients as
\begin{align*}
    \|\nabla_\beta \log p(\beta|X,Y) - \nabla_\beta \log \tilde p(\beta|X,Y)\|_2 &= \|\bar U\bar U^\top X^\top \phip+ UU^\top X^\top \big[(X \bar U\bar U^\top\beta)\circ \phipp\big]\|_2 \\
    &\le \|\bar U\bar U^\top X^\top \phip\|_2 + \| UU^\top X^\top \big[(X \bar U\bar U^\top\beta)\circ \phipp\big]\|_2 \\
    &\le \bar \lambda_1 \|\phip\|_2 + \lambda_1 \| (X \bar U\bar U^\top\beta)\circ \phipp\|_2 \\
    &\le \bar \lambda_1 \|\phip\|_2 + \lambda_1 \bar \lambda_1 \| \bar U^\top\beta\|_2 \| \phipp\|_\infty \\
    &= \bar \lambda_1 \big( \|\phip\|_2 + \lambda_1 \| \bar U^\top\beta\|_2 \| \phipp\|_\infty\big) ,
\end{align*}
where we have written $\phip$ and $\phipp$ in place of $\phip(Y, X\beta)$ and $\phipp(Y, A)$, respectively, for brevity despite their dependence on $\beta$.

Next, let $\alpha$ be the strong log-concavity parameter of $p(\beta|X, Y)$.  \Cref{lemm:strong-convexity} then implies that 
$$
\| \hat \mu  - \bar \mu \|_2 \le \frac{\bar \lambda_1 \big( \|\phip(Y, X\hat \mu)\|_2 + \lambda_1 \| \bar U^\top\hat \mu\|_2 \| \phipp(Y, A)\|_\infty\big)}{\alpha}
$$
as desired, where for each $n\in [N]$, $A_n \in [x_n^\top UU^\top \hat \mu, x_n^\top \hat \mu]$.

\subsection{Bounds on derivatives of higher order for the log-likelihood in logistic regression and other GLMs}\label{sec:log_reg}
We here provide some additional support for the claim that in \Cref{rem:common_glm_derivatives} that the higher order derivatives of the log-likelihood function, $\phi$, are well-behaved.
For logistic regression (which we explore in detail below), for any $y$  in $\{-1, 1\}$ and $a$ in $\R$, it holds that $|\frac{\partial}{\partial a}\phi(y, a)|\le 1$ and $|\frac{\partial^{2}}{\partial^{2} a}\phi(y, a)| \le \frac{1}{4}$.  For Poisson regression with $\phi(y,a) = \log \mathrm{Pois}\big(y|\lambda=\log(1+\exp\{a \})\big)$, both $|\frac{\partial}{\partial a}(y,a)|$ and $|\frac{\partial^{2}}{\partial^{2} a}\phi(y, a)|$ are bounded by a small constant factor of $y$.   Additionally, in these cases $|\frac{\partial^3}{\partial a^3} \phi(y, a)|$ is also well behaved, a fact relevant to \Cref{cor:new_glm_w2_bound}.
However, for alternative mapping functions for Poisson regression, e.g. defining $\E[y_i|x_i,\beta ]=\mathrm{exp}\{x_i^\top \beta \}$, these derivatives will grow exponentially quickly with $x_i^\top \beta$, which illustrates that our provided bounds are sensitive to the particular form chosen for the GLM likelihood.

We now move to compute explicit upper bounds on the derivatives of the log likelihood in logistic regression.
This produces the constants mentioned above, and permits easy computation of upper bounds on the bounds on the approximation error of \methodnamelap provided in \Cref{thm:glm_mean_bound} and \Cref{cor:new_glm_w2_bound}.
In particular the logistic regression mapping function \citep{Huggins2017a} is given as
\begin{align}\label{eqn:logistic_mapping_fcn}
  \phi(y_n, x_n^\top\beta) = -\log \big( 1+\mathrm{exp}\{-y_n x_n^\top\beta\}\big)  ,
\end{align}
where each $y_n \in \{-1, 1\}$.

The first three derivatives of this mapping function and bounds on their absolute values are as follows:
\[
    \label{eqn:first_derivative}
    \phi^\prime(y_n, x_n^\top \beta) :=  \frac{d}{da} \phi(y_n,a)\big|_{a=x_n^\top \beta} =y_n\frac{\mathrm{exp}\{-y x_n^\top \beta\}}{1+\mathrm{exp}\{-y_n x_n^\top \beta\}}
\]
Notably, $\forall a \in \R, y \in \{ -1, 1\}$, $|\phi^\prime(y,a)| < 1$ and

\begin{align}\label{eqn:logistic_curvature}
    \phi^{\prime\prime} (y_n, x_n^\top \beta)  :&= \frac{d^2}{da^2} \phi(y,a)\big|_{a=x_n^\top \beta} = -(1+\mathrm{exp}\{ x_n^\top \beta\})\inv(1+\mathrm{exp}\{-x_n^\top \beta\})\inv.
\end{align}

Furthermore, for any $a$ in $\R$ and $y$ in $\{ -1, 1\}$, $-\frac{1}{4}\le \phi^{\prime\prime}(y,a) < 0$. This implies that the Hessian of the negative log likelihood will be positive semi-definite everywhere.  We additionally have

\begin{equation}\label{eqn:logistic_jerk}
\begin{split}
    \frac{d^3}{da^3}\phi(y,a) = \phi^{\prime\prime\prime} (a)= \frac{\big(\mathrm{exp}\{a\}(\mathrm{exp}(-a) - 1)\big)}{(1+\mathrm{exp}\{a\})^3} \\
\end{split}
\end{equation}
which for any $a$ in $\R$ satisfies, $-\frac{1}{6\sqrt{3}} \leq \phi^{\prime\prime\prime}(a)\leq \frac{1}{6\sqrt{3}}$.

\subsection{Asymptotic inconsistency of the approximate posterior mean within the span of the projections}\label{sec:glm_not_consistent}
Consider a Bayesian logistic regression, in which 
\(
x_i &\sim \mathcal{N}\Big(\begin{bmatrix}0\\0\end{bmatrix},
    \begin{bmatrix}1 &0 \\0 &0.99 \end{bmatrix}\Big), &
\beta &= \begin{bmatrix}10\\1000\end{bmatrix}, &
y_i &\sim \textrm{Bern}\big( (1+\exp\{x_i^\top\beta \})\inv\big).
    \)
In this setting, a rank 1 approximation of the design will capture only the first dimension of data (i.e. $U_N \rightarrow U_*=[1,0]$). 
However the second dimension explains almost all of the variance in the responses. 
As such $y_i| U_*^\top x_i, \beta \overset{d}{\approx} \textrm{Bern}(1/2)$ and we will get $U_*^\top \beta| X, Y  = \beta_1 | X, Y\approx 0.0$  under  $\tilde p$.

\subsection{Proof of \Cref{cor:new_glm_w2_bound}}

Our proof proceeds via an upper bound on the $(2, \hat p)$-Fisher distance between $\hat p$ and $\bar p$ \citep{huggins2018nonasymptotic}. 
Specifically, the $(2, \hat p)$-Fisher distance given by
\begin{equation}\label{eqn:fisher_dist}
    d_{2,\hat p}(\hat p, \bar p) = \left( \int \| \nabla_{\beta} \log \hat p(\beta) - \nabla_\beta \log \bar p(\beta) \|_2^2 d p(\beta) \right)^{\frac{1}{2}}. 
\end{equation}

Given the strong log-concavity of $\bar p$, our upper bound on this Fisher distance immediately provides an upper-bound on the $2$-Wasserstein distance \citep{huggins2018nonasymptotic}.

We first recall that $\hat p$ and $\bar p$ are defined by Laplace approximations of $\tilde p(\beta|X, Y)$ and $p(\beta|X,Y)$ respectively.  As such we have that

\begin{align}
  \nonumber
    \log \hat p(\beta)\stackrel{c}{=} -\frac{1}{2}(\beta-\hat \mu)^\top \big(\Sigma_\beta\inv - UU^\top X^\top \diag(\phipp(Y, XUU^\top \hat \mu)) X UU^\top\big) ( \beta - \hat \mu)
\end{align}
    where $\phipp(Y, XUU^\top \hat \mu)$ is defined as in Algorithm 1 such that $\phipp(Y, X \beta)_i=\frac{d^2}{da^2} \log p(y_i|x^\top \beta=a)|_{a=x_i^\top\beta}$, and 
\begin{align}
  \nonumber
    \log \bar p(\beta)\stackrel{c}{=} -\frac{1}{2}(\beta-\bar \mu)^\top \big(\Sigma_\beta\inv - X^\top \diag(\phipp(Y, X\bar \mu)) X \big) ( \beta - \bar \mu).
\end{align}

Accordingly, 
\begin{align}
  \nonumber
    \nabla_\beta \log \hat p(\beta) = -(\beta- \hat \mu)^\top \big(\Sigma_\beta\inv - UU^\top X^\top \diag(\phipp(Y, XUU^\top \hat \mu))XUU^\top\big)
\end{align}
and 
\begin{align}
  \nonumber
    \nabla_\beta \log \bar p(\beta) = -(\beta- \bar \mu)^\top \big[\Sigma_\beta\inv - X^\top \diag(\phipp(Y, X\bar \mu))X\big]
\end{align}

To define an upper bound on $d_{2,\hat p}(\hat p, p) $, we must consider the difference between the gradients,
\begin{align}
  \nonumber
    \nabla_\beta \log \hat p(\beta) - \nabla_\beta \log \bar p(\beta)  =  & -(\beta- \hat \mu)^\top \big\{ \Sigma_\beta\inv - UU^\top X^\top \diag[\phipp(Y, XUU^\top \hat \mu)] XUU^\top \big\} \\
  \nonumber
    &+ (\beta - \bar \mu)^\top \big\{ \Sigma_\beta\inv - X^\top \diag[\phipp(Y, X\bar \mu)]X \big\} \\
  \nonumber
    &=(\hat \mu - \bar \mu)\Sigma_\beta\inv + (\beta -\hat \mu)^\top UU^\top X^\top \diag[\phipp(Y, XUU^\top \hat \mu)]XUU^\top \\
  \nonumber
    &-(\beta - \bar \mu)^\top X^\top \diag[\phipp(Y, X\bar \mu)]X .
\end{align}

Appealing to Taylor's theorem, we can rewrite $\phipp(Y, XUU^\top \hat \mu)$ as 

\begin{align}
  \nonumber
    \phipp(Y, XUU^\top \hat \mu) &= \phipp(Y, X \bar \mu) + (XUU^\top \hat \mu - X\bar \mu) \circ \phippp(Y, A)\\
  \nonumber
    &= \phipp(Y, X\bar \mu)+(XUU^\top \hat \mu-X\hat \mu + X(\hat \mu -\bar \mu)) \circ \phippp(Y, A) \\
  \nonumber
    &= \phipp(Y, X\bar \mu) - X \bar U \bar U^\top \circ \phippp(Y, A) + X(\hat \mu - \bar \mu) \circ \phippp(Y, A)\\
  \nonumber
    &= \phipp(Y, X\bar \mu) + R ,
\end{align} 

where the first line follows from Taylor's theorem by appropriately choosing each $A_i \in [x_i^\top UU^\top \hat \mu, x_i^\top \bar \mu]$, and in the fourth line we substitute in $R:=-X \bar U \bar U^\top \circ \phippp(Y, A) + X(\hat \mu - \bar \mu) \circ \phippp(Y, A)$.

    We now can rewrite the difference in the gradients as
\begin{align}
  \nonumber
    \nabla_\beta \log \hat p(\beta) - \nabla_\beta \log \bar p(\beta) &= (\hat \mu -\bar \mu)\Sigma_\beta\inv \\
  \nonumber
    &+ (\beta -\hat \mu)^\top UU^\top X^\top\diag[\phipp(Y, X\bar \mu)]XUU^\top \\
  \nonumber
    &+ (\beta - \hat \mu)^\top UU^\top X^\top \diag(R) XUU^\top \\
  \nonumber
    &- (\beta - \bar \mu)^\top X^\top \diag(\phipp(Y, X\bar \mu))X\\
  \nonumber
    &= (\hat \mu -\bar \mu)^\top(\Sigma_\beta\inv  - UU^\top X^\top \diag[\phipp(Y, X\bar \mu)]XUU^\top) \\
  \nonumber
    &+(\beta - \hat \mu)^\top UU^\top X^\top \diag(R) XUU^\top \\
  \nonumber
    &- (\beta - \bar \mu)^\top UU^\top X^\top  \diag(\phipp(Y, X\bar \mu)) X\bar U \bar U^\top \\
  \nonumber
    &- (\beta - \bar \mu)^\top \bar U\bar U^\top X^\top  \diag(\phipp(Y, X\bar \mu)) X U U^\top \\
  \nonumber
    &- (\beta - \bar \mu)^\top \bar U\bar U^\top X^\top  \diag(\phipp(Y, X\bar \mu)) X\bar U \bar U^\top  .
\end{align}

Which is obtained by first writing $ X^\top \diag[\phipp(Y, X\bar \mu)]X$ in the fourth line as $ (UU^\top X^\top + \bar U \bar U^\top X^\top) \diag[\phipp(Y, X\bar \mu)](XUU^\top  + X \bar U \bar U^\top)$, multiplying through and rearranging the resulting terms.

Given this form of the difference in the gradients, we may upper bound its norm as
\begin{align}
  \nonumber
    \|\nabla_\beta \log \hat p(\beta) - \nabla_\beta \log \bar p(\beta)\|_2 &\le \|\hat \mu - \bar \mu\|_2 \|\Sigma_\beta\inv - UU^\top X^\top \diag[\phipp(Y, X\bar \mu)]XUU^\top\|_2 \\
  \nonumber
    &\phantom{blocked out}+ \| \beta- \hat \mu \|_2\|UU^\top X^\top \diag(R)XUU^\top\|_2 \\
  \nonumber
    &\phantom{blocked out}+ \|\beta - \bar \mu \|_2 \|UU^\top X^\top \diag[\phipp(Y, X\bar \mu)]X \bar U\bar U^\top + \\  
  \nonumber
    &\phantom{blocked out more} \bar U\bar U^\top X^\top\diag[\phipp(Y, X\bar \mu)] X UU^\top +  \bar U\bar U^\top X^\top\diag[\phipp(Y, X\bar \mu)] X\bar U\bar U^\top  \|_2\\
  \nonumber
    &\le \|\hat \mu - \bar \mu\|_2 \| \Sigma_\beta\inv-UU^\top X^\top \diag[\phipp(Y, X\bar \mu)]XUU^\top\|_2 \\
  \nonumber
    &\phantom{blocked out}+ \| \beta- \hat \mu \|_2\|UU^\top X^\top \diag(R)XUU^\top\|_2 \\
  \nonumber
    &\phantom{blocked out}+ \|\beta - \bar \mu \|_2 \big\{\|\bar U\bar U^\top X^\top\diag[\phipp(Y, X\bar \mu)] X \bar U\bar U^\top\|_2 +  \\
  \nonumber
    &\phantom{blocked out more}2\|\bar U\bar U^\top X^\top\diag[\phipp(Y, X\bar \mu)]X UU^\top\|_2\big\}\\
   \nonumber
   &\textrm{by the triangle inequality.} \\
  \nonumber
    &\le \|\hat \mu - \bar \mu\|_2 \left\{\| \Sigma_\beta\inv\|_2 +\| U\diag(\lambda) V^\top\|_2  \| \diag[\phipp(Y, X\bar \mu)]\|_2 \| V \diag(\lambda) U^\top\|_2\right\} \\
  \nonumber
    &\phantom{blocked out}+ \| \beta- \hat \mu \|_2\|U \diag(\lambda) V^\top \|_2 \| \diag(R)\|_2 \|V \diag(\lambda) U^\top\|_2 \\
  \nonumber
    &\phantom{blocked out}+ \|\beta - \bar \mu \|_2 \big\{\|\bar U  \diag(\bar \lambda) \bar V^\top\|_2 \|\diag[\phipp(Y, X\bar \mu)] \|_2 \| \bar V \diag(\bar \lambda) \bar U^\top\|_2 +  \\
  \nonumber
    &\phantom{blocked out more}2\|\bar U^\top \diag(\bar \lambda) \bar V^\top \|_2 \| \diag[\phipp(Y, X\bar \mu)]\|_2 \| V \diag(\lambda) U^\top\|_2\big\}\\
  \nonumber
    &\textrm{by again using the triangle inequality, and decomposing $X^\top$ into $U\diag(\lambda) V^\top + \bar U \diag(\bar \lambda) \bar V^\top$.} \\
  \nonumber
    &\le \|\hat \mu - \bar \mu\|_2 \big(\| \Sigma_\beta\inv\|_2+\lambda_1^2 \|\phipp\|_\infty\big)
    + \lambda_1^2\| \beta- \hat \mu \|_2 \|R\|_\infty
    + (\bar \lambda_1^2 + 2 \lambda_1 \bar \lambda_1) \|\beta - \bar \mu \|_2 \|\phipp\|_2 ,
\end{align}
where in the last line we have shortened $\phipp(Y, X\bar \mu)$ to $\phipp$ for convenience.

Next noting that $\| \bar \mu - \hat \mu\|_2 \le \bar \lambda_1 c$ for $c:= \frac{\|\phip(Y, X\hat \mu)\|_2 + \lambda_1 \| \bar U^\top\hat \mu\|_2 \| \phipp(Y, A)\|_\infty}{\alpha}$, where $\alpha$ is the strong log concavity parameter of $p(\beta|X, Y)$ (which follows from \Cref{thm:glm_mean_bound}), we can see that $\| R \|_\infty \le \bar \lambda_1 r$ where $r:= (\| U^\top \hat \mu\|_\infty \|\phippp(Y,A)\|_\infty + \lambda_1 c\|\phippp(Y, A)\|_\infty)$.  That $r$ is bounded follows from the assumption that $\log p(y|x, \beta)$ has bounded third derivatives, an equivalent to a Lipschitz condition on $\phi^{\prime \prime}$.   We can next simplify this upper bound to
\begin{align}
  \nonumber
    \|\nabla_\beta \log \hat p(\beta) - \nabla_\beta \log \bar p(\beta)\|_2 &\le
    \bar\lambda_1 c [\|\Sigma_\beta\inv\|_2 + \lambda_1^2\|\phipp\|_\infty] + 
    \lambda_1^2 \bar\lambda_1 r \| \beta - \hat \mu \|_2 + 
    \bar \lambda_1 (\bar \lambda_1 +2 \lambda_1)\|\beta - \bar \mu \|_2 \|\phipp\|_\infty \\
  \nonumber
    &= \bar\lambda_1 \big[ 
    c(\|\Sigma_\beta\inv\|_2 + \lambda_1^2\|\phipp\|_\infty) + 
    \lambda_1^2 r \| \beta - \hat \mu \|_2 + 
    (\bar \lambda_1 +2 \lambda_1)\|\beta - \bar \mu \|_2 \|\phipp\|_\infty \big]\\
  \nonumber
    &\le \bar\lambda_1 \big[ 
    c(\|\Sigma_\beta\inv\|_2 + \lambda_1^2\|\phipp\|_\infty) + 
    \lambda_1^2 r \| \beta - \hat \mu \|_2 + 
    (\bar \lambda_1 +2 \lambda_1)(\|\hat \mu - \bar \mu\|_2 + \|\beta - \hat \mu \|_2) \|\phipp\|_\infty \big]\\
   \nonumber
   &\textrm{by the triangle inequality.} \\
  \nonumber
    &\le \bar\lambda_1 \big[ 
    c(\|\Sigma_\beta\inv\|_2 + \lambda_1^2\|\phipp\|_\infty) + 
    \lambda_1^2 r \| \beta - \hat \mu \|_2 + 
    (\bar \lambda_1 +2 \lambda_1)(\bar \lambda_1 c + \|\beta - \hat \mu \|_2) \|\phipp\|_\infty \big]\\
  \nonumber
    &= \bar\lambda_1 \big[ 
    c(\|\Sigma_\beta\inv\|_2 + \lambda_1^2\|\phipp\|_\infty) +
     c  (\bar \lambda_1^2 +2 \lambda_1\bar\lambda_1)\|\phipp\|_\infty 
    + (\lambda_1^2 r +  (\bar \lambda_1 +2 \lambda_1)\|\phipp\|_\infty) \|\beta - \hat \mu \|_2 \big]\\
  \nonumber
    &= \bar\lambda_1 \big[ 
    c (\|\Sigma_\beta\inv\|_2 + (\lambda_1 + \bar \lambda_1)^2 \|\phipp\|_\infty) +
    (\lambda_1^2 r +  (\bar \lambda_1 +2 \lambda_1)\|\phipp\|_\infty) \|\beta - \hat \mu \|_2 \big] .
\end{align}

Thus, taking the expectation of this upper bound on the norm squared over $\beta$ with respect to $\hat p$ we get
\begin{align}
  \nonumber
    d_{2,\hat p}^2(\hat p, p) &\le \E_{\hat p(\beta)}\left( 
    \bar\lambda_1^2 \left\{ 
    c \left[ \|\Sigma_\beta\inv\|_2 + (\lambda_1 + \bar \lambda_1)^2 \|\phipp\|_\infty\right] +
    \left[\lambda_1^2 r +  (\bar \lambda_1 +2 \lambda_1)\|\phipp\|_\infty\right] \|\beta - \hat \mu \|_2 \right\}^2
    \right)\\
  \nonumber
    &\le 2 \bar \lambda_1^2 \E_{\hat p(\beta)}\left\{ 
    c^2 \left[\|\Sigma_\beta\inv\|_2 + (\lambda_1 + \bar \lambda_1)^2 \|\phipp\|_\infty\right]^2
    + \left[ \lambda_1^2 r +  (\bar \lambda_1 +2 \lambda_1)\|\phipp\|_\infty\right]^2 \|\beta - \hat \mu \|_2^2
    \right\}\\
   \nonumber
    &\textrm{since $\forall a, b \in \mathbb{R}, (a+b)^2 \le 2(a^2+b^2)$} \\
  \nonumber
    &= 2 \bar \lambda_1^2 \left\{ 
    c^2 \left[\|\Sigma_\beta\inv\|_2 + (\lambda_1 + \bar \lambda_1)^2 \|\phipp\|_\infty\right]^2
    + \left[\lambda_1^2 r +  (\bar \lambda_1 +2 \lambda_1)\|\phipp\|_\infty\right]^2 \E_{\hat p(\beta)}[\|\beta - \hat \mu \|_2^2]
    \right\}\\
  \nonumber
    &= 2 \bar \lambda_1^2 \left\{ 
    c^2 \left[\|\Sigma_\beta\inv\|_2 + (\lambda_1 + \bar \lambda_1)^2 \|\phipp\|_\infty\right]^2
    + \left[\lambda_1^2 r +  (\bar \lambda_1 +2 \lambda_1)\|\phipp\|_\infty\right]^2 \trace(\hat \Sigma)
    \right\} .
\end{align}

Next noting that $\bar p$ is strongly $\|\bar \Sigma\|_2\inv$ log-concave, we may apply \Cref{lemm:fisher_W2_bound}, stated below,
to obtain that 
\begin{align}
  \nonumber
    W_2(\hat p, \bar p) &\le \| \bar \Sigma \|_2 \sqrt{
        2 \bar \lambda_1^2 \left\{
        c^2 \left[\|\Sigma_\beta\inv\|_2 + (\lambda_1 + \bar \lambda_1)^2 \|\phipp\|_\infty\right]^2
    + \left[\lambda_1^2 r +  (\bar \lambda_1 +2 \lambda_1)\|\phipp\|_\infty\right]^2 \trace(\hat \Sigma)
    \right\} }\\
  \nonumber
    &\le \sqrt{2} \bar \lambda_1 \| \bar \Sigma \|_2 
    \left\{ 
    c \left[\|\Sigma_\beta\inv\|_2 + (\lambda_1 + \bar \lambda_1)^2 \|\phipp\|_\infty\right]
    + \left[\lambda_1^2 r +  (\bar \lambda_1 +2 \lambda_1)\|\phipp\|_\infty\right] \sqrt{\trace(\hat \Sigma)}
    \right\} ,
\end{align}
which is our desired upper bound.


\begin{theorem}\label{lemm:fisher_W2_bound}
Suppose that $p(\beta)$ and $q(\beta)$ are twice continuously differentiable and that $q$ is $\alpha$-strongly log concave.  Then
$$W_2(p, q) \le \alpha \inv d_{2, p}(p, q),$$
where $W_2$ denotes the $2$-Wasserstein distance between $p$ and $q$.
\end{theorem}
\begin{proof}
    This follows from \citet{huggins2018nonasymptotic} Theorem 5.2, or similarly from \citet{bolley2012convergence} Lemma 3.3 and Proposition 3.10.
\end{proof}

\subsection{Proof of bounded asymptotic error}\label{sec:glm_bounded_asymptotic_error}
We here provide a formal statement and proof of \Cref{thm:asymptotic_glm_mean_bound_concise}, detailing the required regularity conditions.

\begin{theorem}[Asymptotic]\label{thm:asymptotic_glm_mean_bound}
Assume $x_i \overset{\mathrm{i.i.d.}}{\sim} p_*$ for some distribution $p_*$ such that $\E_{p_*}[x_i x_i^\top]$ exists and is non-singular with diagonalization 
$\E_{p_*}[x_i x_i^\top]= U_*^\top \diag(\lambda) U_* + \bar U_*^\top \diag(\bar \lambda) \bar U_*$ such that 
    $\mathrm{min}(\lambda)>\mathrm{max}(\bar \lambda)$.  Additionally, for a strictly concave (in its second argument), twice differentiable log-likelihood function $\phi$ with bounded second derivatives (in both arguments) and some $\beta\in \R^D$, let $y_i | x_i \sim \exp\{\phi(y_i, x_i^T\beta) \}$. 
Also, suppose that $\E \|y_i\|_2^2 <\infty$.
Then if $p(\beta)$ is log-concave and positive on $\mathbb{R}^D$, the asymptotic error (in $N$) of the exact relative to approximate maximum a posteriori parameters,
$\hat \mu = \lim_{N\rightarrow \infty} \hat \mu_N$ and $\bar \mu = \lim_{N\rightarrow \infty} \bar \mu_N$ is finite
(where $\hat \mu_N$ and $\bar \mu_N$ are the approximate and exact MAP estimates, respectively, after $N$ data-points),
i.e., $\lim_{n\rightarrow \infty}\|\hat \mu_N -\bar \mu_N\|$ exists and is finite.
\end{theorem}

\begin{proof}
Before beginning, let $\pr$ denote a Borel probability measure on the sample space on which our random variables, $\{x_i\}$ and $\{y_i\}$, are defined such that these random variables are distributed as assumed according to $\pr$.
In what follows we demonstrate the asymptotic error is finite $\pr$-almost surely.
To this end, it suffices to show that $\hat \mu_N \overset{a.s.}{\rightarrow} \hat \mu$ and   $\bar \mu_N \overset{a.s.}{\rightarrow} \bar \mu$ for some $\hat \mu, \bar \mu$ in $\mathbb{R}^D$.

\subsection*{Strong convergence of the exact MAP ($\bar \mu_N \stackrel{a.s.}{\rightarrow} \bar \mu$)}
This follows from Doob's consistency theorem \citep[Theorem 10.10]{van2000asymptotic}. The only nuance required in the application of this theorem here is that we must accommodate the regression setting.
However by constructing a single measure $\pr$ governing both the covariates and responses, this simply becomes a special case of the usual theorem for unconditional models.

\subsection*{Strong convergence of the approximate MAP ($\hat \mu_N \stackrel{a.s.}{\rightarrow} \hat \mu$)}
In contrast to the strong consistency of $\bar \mu_N$, showing convergence of $\hat\mu_N$ requires more work.
This is because we cannot rely on standard results such as Bernstein--Von Mises or Doob's consistency theorem, which require correct model specification.  
Since we have introduced the likelihood approximation $\tilde p(y|x,\beta) \neq p(y|x,\beta)$, the vector $\hat \mu_N$ is the MAP estimate under a misspecified model.

We demonstrate almost sure convergence  in two steps; first we show that $U_*^\top \hat \mu_N$ converges almost surely to some $\gamma^* \in \R^M$; 
then we show that $\hat \mu_N = U_* U_*^\top \hat \mu_N + \bar U_N \bar U_N^\top \hat \mu_N$ must converge as a result. 
Since $U_N U_N^\top \overset{a.s.}{\rightarrow} U_* U_*^\top$ (as follows from entry-wise almost sure convergence of 
$\frac{1}{N} X^\top X \rightarrow \E_{p_*}[x_i x_i^\top]$ and the Davis--Kahan Theorem \citep{davis1970rotation}), 
this guarantees strong convergence of $\hat \mu_N =U_N U_N^\top \hat \mu_N + \bar U_N \bar U_N^\top \hat \mu_N$.

\textbf{Part I: strong convergence of the projected approximate MAP, $U_*\hat \mu_N \stackrel{a.s.}{\rightarrow} \gamma^*$}

Let $U_*\in \R^{D,M}$ be the top $M$ eigenvectors of $\E_{p_*}[x_ix_i^\top]$, and recall that by assumption for any $y, \phi(y, x_i^T\beta)$ is a strictly concave function of $x_i^\top \beta$, 
    in the sense that for any $y$ and any $b,b^\prime$ in $\R$ and $t$ in $(0,1)$ with $b \ne b^\prime$, $\phi(y,tb+(1-t)b^\prime) > t\phi(y,b)+(1-t)\phi(y,b^\prime)$.  
    Then by \Cref{lemma:unique_max} we have that there is a unique maximizer $\gamma^* = \argmax_{\gamma \in R^M} \E[\phi(y, x^\top U_*\gamma)]$

We next note that the Hessian of the expected approximate negative log likelihood with respect to $\gamma$ is positive definite everywhere,
$$
\nabla_\gamma^2 -\E_{y\sim p(y|x,\beta),x\sim p_*}[\phi(y, x^\top U_*\gamma)]=-\E[\big( \nabla_\gamma \phi^\prime(y, x^\top U_*\gamma) \big)x^\top U_*]=-U_*^\top \E[x\phi^{\prime\prime}(y, x^\top U_*\gamma)x^\top ]U_*\succ 0
$$
    since the strict log concavity and twice differentiability of $\phi$ ensure that $-\E[x\phi^{\prime\prime}(y,x^\top U_*\gamma)x^\top] \succ 0$.

Now consider any compact neighborhood $K\subset \R^M$ containing $\gamma^*$ as an interior point.  
Then, by \Cref{lemma:P_GC_gamma}  the set $\mathcal{F} =\{f_\gamma: X \times Y \rightarrow \R,  (x, y)\mapsto \phi(y, x^\top U_* \gamma) | \gamma \in K \}$ is $\pr$-Glivenko--Cantelli.  
As such $\sup_{f_\gamma \in \mathcal{F}} | \frac{1}{N} \sum_{i=1}^N f_\gamma(x_i, y_i) - \E[f_\gamma(x_i, y_i)]| \overset{a.s.}{\rightarrow} 0$, 
that is to say, the empirical average log-likelihood converges uniformly to its expectation across all $\gamma \in K$. 
    As a result, we have that for $\gamma_N := \argmax_{\gamma \in K} \log \tilde p (U_*\beta = \gamma|X, Y) = \argmax_{\gamma \in K} \frac{1}{N} \big[\log p(U_*^T \beta=\gamma) + \sum_{i=1}^N \phi(y_i, x_i^\top U_*\gamma)\big]$, $\gamma_N\stackrel{a.s.}{\rightarrow} \gamma^*$.

It remains in this part only to show that convergence of the approximate MAP parameter within this subset $K$ implies convergence of $U_*^\top \hat \mu$,
the approximate MAP parameter (across all of $\R^M$).  However, this follows immediately from the strict log concavity of the posterior; 
because $\gamma^* \in K^\circ$, for $N$ large enough each $\gamma_N \in K^\circ$ and we may construct a sub-level set such that $\gamma_N \in C_N \subset K$
such that $\forall \gamma \notin C_N, \log p(U_*^\top\beta =\gamma) + \sum_{i=1}^N \phi(y_i, x_i^\top U_*\gamma) < \log p(U_*^\top\beta =\gamma_N) + \sum_{i=1}^N \phi(y_i, x_i^\top U_*\gamma_N)$.

\textbf{Part II: convergence of $\bar U_* \bar U_*^\top \hat \mu_N + U_*\gamma$}

Using the result of Part I,  we can write that $\hat \mu_N = U_*U_*^\top \hat \mu_N + \bar U_* \bar U_*^\top \hat \mu_N \rightarrow U_* \gamma^* +  \bar U_* \bar U_*^\top \hat \mu_N$. 
However, since $\bar U_* \bar U_*^\top \beta \perp X,Y | U_*^\top \beta$ under $\pr$, convergence of $U_*^\top \hat \mu_N \rightarrow \gamma^*$ 
implies convergence of $\argmax_{\bar U_*^\top \beta} \tilde p( \bar U_*^\top \beta | U_*^\top \beta = U_*^\top \hat \mu_N, X, Y) = \argmax_{\bar U_*^\top \beta} \tilde p( \bar U_*^\top \beta | U_*^\top \beta=U_*^\top)$ 
    to some $\bar U_*^\top \hat \mu_N$ since continuity of $p(\beta)$ and $\tilde p(Y|X, \beta)$ imply continuity of the arg-max.  
    Thus both $\hat \mu_N$ and $\bar \mu_N$ converge, guaranteeing convergence of the asymptotic error.
\end{proof}

\begin{lemma}\label{lemma:unique_max}
    For any $\phi(\cdot, \cdot)$ which is strictly concave in its second argument, if there is a global maximizer $\beta^*=\argmax_{\beta\in\R^D} V(\beta)=\E_{x\sim p_*, y\sim p(y|x,\beta)}[\phi(y, x^\top \beta)]$, then there is a unique global maximizer, 
    $$\gamma^*=\argmax_{\gamma\in\R^M} V(U_*\gamma)$$
\end{lemma}
\begin{proof}
    We first note that $V(\cdot)$ must have bounded sub-level sets.  Thus $W(\cdot):=V(U_*\cdot)$ must also have bounded sub-level sets 
    since $V^{-1}([a,\infty])=\{\beta | V(\beta)\ge a\}\supset \{ \beta|\exists \gamma \in \R^M\  s.t.\  \beta = U_*\gamma\  and \  V(U_*\gamma)\ge a \}=U_* W^{-1}([a, \infty])$.  
    Thus, since $W$ is strictly concave and has bounded sub-level sets, it has a unique maximizer.
\end{proof}

\begin{lemma}\label{lemma:P_GC_gamma}
    Let $K\subset \R^M$ be compact and denote by $X$ and $Y$ the domains of the covariates and responses, respectively. Then under the assumptions of \Cref{thm:asymptotic_glm_mean_bound}, the set $\mathcal{F} =\{f_\gamma: X \times Y \rightarrow \R,  (x, y)\mapsto \phi(y, x^\top U \gamma) | \gamma \in K \}$ is $\pr$-Glivenko--Cantelli.
\end{lemma} 
\begin{proof} 
This result follows from Theorem $19.4$ in \cite{van2000asymptotic}, and builds from example $19.7$ of the same reference;
in particular, the condition of bounded second derivatives of $\phi$ implies that for any $f_\gamma, f_{\gamma^\prime}$ in $\mathcal{F}$ and 
$x$ in $X$, $y$ in $Y$, we have $|f_\gamma(x,y) - f_{\gamma^\prime}(x,y)| \le C \|x\|_2^2$.  
The previous condition is sufficient to ensure finite bracketing numbers, and the result follows. 
Notably, in keeping with example $19.7$ we have that for all $x, y$ and for all $\gamma$ and $\gamma^\prime$ in $K$, 
\begin{equation}
\begin{split}
    | f_\gamma(x, y) - f_{\gamma^\prime}(x,y)| &= | \int_{x^\top U\gamma^\prime}^{x^\top U\gamma} \phi^\prime(y, a)da| \\
    &= | \int_{x^\top U\gamma^\prime}^{x^\top U\gamma} \phi^\prime(y, x^\top U\gamma^\prime)  + 
       \int_{x^\top U\gamma^\prime}^{a} \phi^{\prime\prime}(y, b)db\ da| \\
    &\le \| x^\top U(\gamma -\gamma^\prime) \phi^{\prime}(y, x^\top U \gamma^\prime) \|_2 +
       \frac{1}{2} \|x^\top U(\gamma -\gamma^\prime)\|_2^2 \sup_{a\in \R}\phi^{\prime\prime}(y,a)\\
    &\le \| x^\top U\|_2 (\|y\|_2 + \|x^\top U\|_2 \| \gamma^\prime \|_2 ) \phi^{\prime \prime}_{\mathrm{max}} \|\gamma -\gamma^\prime\|_2 +
       \frac{1}{2} \|x^\top U\|_2^2 \|\gamma -\gamma^\prime\|_2^2 \phi^{\prime \prime}_{\mathrm{max}}\\
    &\le \big[ \frac{3}{2} \| x^\top U\|_2^2 \mathrm{diam}(K) \phi^{\prime \prime}_{\mathrm{max}} + 
       \|x^\top U\|_2 \|y\|_2 \phi^{\prime \prime}_{\mathrm{max}} \big] \|\gamma -\gamma^\prime\|_2\\
    &\le C ( \| x^\top U\|_2^2 + \|x^\top U\|_2 \|y\|_2 ) \|\gamma -\gamma^\prime\|_2,
\end{split}
\end{equation}
    where in the first and second lines we use the fundamental theorem of calculus, and in the fourth and fifth lines we rely on the boundedness of the second derivatives 
    of $\phi$ and that the compactness subsets of $\R^M$ implies boundedness.  In the final line $C$ is an absolute constant.

    Finally, we note that $\E_\pr \|x^\top U\|_2^2 <\infty$ since $\E_\pr \|x^\top U\|_2^2 = \E_\pr x^\top U U^\top x  <\E_\pr x^\top x =\mathrm{Tr}(\E_{p_*} x x^\top) < \infty$,
    and by Cauchy Schwartz, $\E_\pr \|x^\top U\|_2 \|y\|_2 \le \sqrt{\E_\pr \|x^\top U\|_2^2 \E_\pr\|y\|_2^2}\le \infty$. 
    This confirms (as in example $19.7$ \cite{van2000asymptotic}) that for all $\epsilon >0$, the $\epsilon$-bracketing number of $\mathcal{F}$ is finite.
    By Theorem 19.4 of \cite{van2000asymptotic}, this proves that $\mathcal{F}$ is $\pr$-Glivenko-Cantelli.
\end{proof} 

\subsection{Factorized Laplace approximations underestimate marginal variances}\label{sec:bad_marginals}
We here illustrate that the factorized Laplace approximation underestimates marginal variances. Consider for simplicity the case of a bivariate Gaussian with 
\(
\Sigma=
 \begin{bmatrix}
  a  & b \\
  b & c
 \end{bmatrix},\ \ \ 
\)
for which the Hessian evaluated anywhere is
\(
 \Sigma\inv =
 \frac{1}{ac-b^2} \begin{bmatrix}
  c & -b \\
  -b & a
 \end{bmatrix}.
\)
Ignoring off diagonal terms and inverting to approximate $\Sigma_N$, as is done by a diagonal Laplace approximation, yields:
\(
 \tilde \Sigma=
 \begin{bmatrix}
     a -\frac{b^2}{c} & 0 \\
     0 & c-\frac{b^2}{a}
 \end{bmatrix} .
\)
This approximation reports marginal variances which are lower than the exact marginal variances.

That this approximation underestimates marginal variances in the more general $D>2$ dimensional case may be easily seen from considering the block matrix inversion of $\Sigma$, with blocks of dimension $1\times1$,\ $(D-1)\times1$,\ $1 \times (D-1)$\ and\ $(D-1) \times (D-1)$, and noting that the Schur complement of a positive definite covariance matrix will always be positive definite.

\section{LR-MCMC} \label{sec:lr_mcmc}

We provide the LR-MCMC algorithm for performing fast MCMC in generalized linear models with low-rank data approximations.

\begin{algorithm*}[!ht]
  \caption{LR-MCMC for Bayesian inference in GLMs with low-rank data approximations.}\label{alg:fast_laplace_mcmc}
\begin{algorithmic}[1]
  \InitThreeCols
  \Phase {{\bfseries Input:} prior $p(\beta)$, data $X \in \R^{N,D}$, rank $M\ll D$, GLM mapping $\phi$, MCMC transition kernel $q(\cdot, \cdot)$, number of MCMC iterations $T$. Time and memory complexities that are not included depend on the specific choice of MCMC transition kernel.}
  \ThreeHeads{Pseudo-Code}{Time Complexity}{Memory Complexity}
  \Phase{Data preprocessing --- $M$-Truncated SVD }
  \LeftMidRight{\State $U, \diag(\mathbf{\lambda}), V := \operatorname{truncated-SVD}(X^\top , M)$}{$O(NDM)$}{$O(NM+ DM)$}
  \LeftMidRight{\State $X_U = X U$}{$O(NM)$}{$O(NDM)$}
  \vspace{.2cm}
  \Phase{Propose $\beta^{(t)} \in \mathbb{R}^D$, compute likelihood}
    \LeftMidRight{\State $\beta^{(t)} \sim q(\beta^{(t)}, \beta^{(t-1)})$}{---}{---} \label{line:transition}
    \LeftMidRight{\State $\mathcal{L}_t := \sum_{i=1}^N \phi(y_i, x_i^\top UU^\top \beta^{(t)}) + \log p(\beta^{(t)})$}{$O(1)$}{$O(NM + MD)$}
  \vspace{.2cm}
  \Phase{Accept or Reject}
  \LeftMidRight{\State Acceptance probability $p_A := \min \left(1, \frac{\mathcal{L}_t}{\mathcal{L}_{t-1}}\right)$}{$O(1)$}{$O(1)$}
  \LeftMidRight{\State Accept $\beta^{(t)}$ with probability $p_A$ }{$O(1)$}{$O(1)$}  
  \vspace{.2cm}
  \Phase{Repeat steps 3-6 for $T$ iterations}
 \end{algorithmic}
\end{algorithm*}
The transition in \Cref{line:transition} may additionally benefit from the \methodname approximation.
In particular, widely used algorithms  such as Hamiltonian Monte Carlo and the No-U-Turn Sampler rely on many $O(ND)$-time likelihood and gradient evaluations, the cost of which can be reduced to $O(NM + DM)$ with \methodname.
An implementation of this approximation is given in the \texttt{Stan} model in \Cref{sec:stan_code} with performance results in \Cref{fig:comp_trade_off_hmc,fig:means_and_variances_HMC}.

\section{LR-Laplace with non-Gaussian priors}\label{sec:lr_laplace_general}
As discussed in the main text, we can maintain computational advantages of \methodname even when we have non-Gaussian priors.  This admits the procedure provided in \Cref{alg:fast_laplace_other_priors}.  
\begin{algorithm*}[!ht]
    \caption{LR-Laplace for Bayesian inference in GLMs with low-rank data approximations and twice differentiable prior. Time and memory complexities which are not included depend on the choice of prior and optimisation method, which can be problem specific.}\label{alg:fast_laplace_other_priors}
\begin{algorithmic}[1]
    \InitThreeCols

    \Phase {{\bfseries Input:} twice differentiable prior $p(\beta)$, data $X \in \R^{N,D}$, rank $M\ll D$, GLM mapping $\phi$ with $\phi^{\prime\prime}$ (see \Cref{eqn:mapping_fcn,sec:fast_laplace_approximations})}
    \ThreeHeads{Pseudo-Code}{Time Complexity}{Memory Complexity}
    \Phase{Data preprocessing --- $M$-Truncated SVD }
    \LeftMidRight{\State $U, \diag(\mathbf{\lambda}), V := \operatorname{truncated-SVD}(X^T, M)$}{$O(NDM)$}{$O(NM+ DM)$}
    \vspace{.2cm}
    \Phase{Optimize to find approximate MAP estimate (in $D$-dimensional space)}
    \LeftMidRight{\State $\hat \mu := \argmax_{\mu \in \R^D} \sum_{i=1}^N \phi(y_i, x_i U U^\top \mu)+\log p(\beta=\mu)$}{---}{---} \label{line:find_MAP}
    \vspace{.2cm}
    \Phase{Compute approximate posterior covariance\footnotemark}
    \LeftMidRight{\State $\hat \Sigma\inv := -\nabla_\beta^2 \log p_\beta(\hat \mu) - U U^\top  X^\top  \diag( \phipp_{\hat \mu})XUU^\top$}{---}{---}
    \LeftMidRight{\State $K := [\nabla_\beta^2 \log p(\beta)|_{\beta=\hat \mu}]\inv$}{---}{---} \label{line:K}
    \LeftMidRight{\State $\hat \Sigma := -K + KU 
        \big(
            [ U^\top  X^\top  \diag( \phipp_{\hat \mu})XU ]\inv + U^\top K U
        \big)\inv U^\top K
    $}{---}{---}
    \vspace{.2cm}
    \Phase{Compute variances and covariances of parameters}
    \LeftMidRight{\State $\mathrm{Var}_{\hat p}(\beta_i)=e_i^\top  \hat \Sigma e_i$}{---}{---}
    \LeftMidRight{\State $\mathrm{Cov}_{\hat p}(\beta_i, \beta_j)=e_i^\top  \hat \Sigma e_j$}{---}{---}
  \end{algorithmic}
\end{algorithm*}
\footnotetext{To keep notation concise we use $\phipp_{\hat \mu}$ to denote $\phipp(Y, XUU^\top\hat \mu)$}

In order for this more general LR-Laplace algorithm to be computationally efficient, we still require that the prior have some properties which can accommodate efficiency.  In particular \Cref{line:K} demands that the Hessian of the prior is computed and inverted, as will true even in the high-dimensional setting when, for example, the prior factorizes across dimensions.  Additionally, properties of the prior such as log concavity will facilitate efficient optimisation in \Cref{line:find_MAP}.

}

\end{document}